\documentclass[onecolumn,12pt]{IEEEtran}
\usepackage{cite}
\usepackage{algorithmic}
\usepackage{float}
\usepackage{array}
\usepackage{bbm}
\usepackage{amsfonts}
\usepackage{graphicx}
\usepackage{xcolor}

\usepackage{mathtools}
\usepackage{amsmath, amsthm, amssymb}
\usepackage[outdir=./]{epstopdf}
\usepackage{dsfont}
\usepackage{mathrsfs}
\newtheorem{theorem}{Theorem}
\newtheorem{Corollary}{Corollary}
\newtheorem{Definition}{Definition}

\usepackage{color,algorithm,algorithmic,amsbsy,amsmath,amssymb,epsfig,bbm,mathrsfs,fancyhdr,fancyvrb,url}
\usepackage{graphicx}
\usepackage{ulem}
\usepackage{subcaption}
\usepackage{afterpage}  %use this before the definition of the figure --> \afterpage{\clearpage}

\usepackage{xcolor,cite,etoolbox}
\makeatletter
\pretocmd\@bibitem{\color{black}\csname keycolor#1\endcsname}{}{\fail}
\newcommand\citecolor[1]{\@namedef{keycolor#1}{\color{black}}}
\makeatother
\citecolor{ji2018ultra}
\citecolor{ericsson_wp}
\citecolor{wu2019low}
\citecolor{lv2018millimeter}
\citecolor{yi2017modeling}
\citecolor{moltchanov2019socially}
\citecolor{ibrahim2019meta}

\newcommand{\beq}{\begin{equation}}
\newcommand{\eeq}{\end{equation}}
\newcommand{\beqn}{\begin{eqnarray}}
\newcommand{\eeqn}{\end{eqnarray}}

\newtheorem{lemma}[theorem]{Lemma}

\begin{document}
\title{The Meta Distributions of the $\mathrm{SIR}/\mathrm{SNR}$ and Data Rate in Coexisting Sub-6GHz and Millimeter-wave Cellular Networks}

\author{Hazem Ibrahim, Hina Tabassum, and Uyen T. Nguyen
        % <-this % stops a space
\IEEEcompsocitemizethanks{\IEEEcompsocthanksitem The authors are with the Department of Electrical Engineering and Computer Science, York University, Toronto, Ontario, M3J 1P3 Canada (e-mail:{hibrahim,hina,utn}@cse.yorku.ca). Part of the manuscript was presented at the 2019 IEEE ICC: SAC Internet of Things Track \cite{ibrahim2019meta}.}% <-this % stops an unwanted space
}

%\markboth{Journal of \LaTeX\ Class Files,~Vol.~13, No.~9, September~2014}%
%{Shell \MakeLowercase{\textit{et al.}}: Bare Demo of IEEEtran.cls for Computer Society Journals}
\maketitle
%\IEEEtitleabstractindextext{

\begin{abstract}
Meta distribution is a  fine-grained unified performance metric that enables us to evaluate  the {reliability and latency}  of next generation wireless networks, in addition to the conventional coverage probability. In this paper, using stochastic geometry tools, we develop a systematic framework to characterize the  meta distributions of the downlink  signal-to-interference-ratio (SIR)/signal-to-noise-ratio (SNR) and data rate of a typical device in a cellular network with coexisting sub-6GHz and millimeter wave (mm-wave) spectrums. Macro base-stations (MBSs) transmit on sub-6GHz channels (which we term ``microwave'' channels), whereas small base-stations (SBSs) communicate with {\color{black}devices} on mm-wave channels. The SBSs are connected to MBSs via a microwave ($\mu$wave) wireless backhaul.  The $\mu$wave channels are interference limited and mm-wave channels are noise limited; therefore, we have the meta-distribution of SIR and SNR in $\mu$wave and mm-wave channels, respectively. To model the line-of-sight (LOS) nature of mm-wave channels, we use Nakagami-m fading model. To derive the meta-distribution of SIR/SNR, we characterize the conditional success probability (CSP) (or equivalently reliability) and its $b^{\mathrm{th}}$ moment  for  a typical {\color{black}device}  (a) when it associates to a $\mu$wave MBS for {\em direct} transmission, and (b) when it associates to a mm-wave SBS  for {\em dual-hop} transmission (backhaul and access transmission).  Performance metrics such as  the mean and variance of the local delay (network jitter),  mean of the CSP (coverage probability), and variance of the CSP are derived. Closed-form expressions are presented for special scenarios. The extensions of the developed framework to the $\mu$wave-only network or mm-wave only networks where SBSs have mm-wave backhauls are discussed. Numerical results validate the analytical results. Insights are extracted related to the reliability, coverage probability, and latency of the considered  network.
\end{abstract}
% Note that keywords are not normally used for peerreview papers.
\begin{IEEEkeywords}
5G Cellular networks, millimeter wave, meta distribution, reliability, latency, wireless backhaul, Nakagami fading, stochastic geometry.
\end{IEEEkeywords}%}
% make the title area
% To allow for easy dual compilation without having to reenter the
% abstract/keywords data, the \IEEEtitleabstractindextext text will
% not be used in maketitle, but will appear (i.e., to be "transported")
% here as \IEEEdisplaynontitleabstractindextext when the compsoc
% or transmag modes are not selected <OR> if conference mode is selected
% - because all conference papers position the abstract like regular
% papers do.
\IEEEdisplaynontitleabstractindextext
% \IEEEdisplaynontitleabstractindextext has no effect when using
% compsoc or transmag under a non-conference mode.
% For peer review papers, you can put extra information on the cover
% page as needed:
% \ifCLASSOPTIONpeerreview
% \begin{center} \bfseries EDICS Category: 3-BBND \end{center}
% \fi
%
% For peerreview papers, this IEEEtran command inserts a page break and
% creates the second title. It will be ignored for other modes.
\IEEEpeerreviewmaketitle

\section{Introduction}\label{sec:introduction_report}

\IEEEPARstart{T}{he sub-6GHz spectrum} is running out of bandwidth to support a huge number of devices in the cellular networks. Therefore, cellular operators of the upcoming 5G networks will tap into the millimeter-wave (mm-wave) spectrum to use wider bandwidths. The mm-wave spectrum has wider bandwidths that can meet higher traffic demands and support data rates into the order of gigabits per second. Mm-wave spectrum usage is one of the key enablers of 5G and beyond networks \cite{NR} and will coexist with sub-6GHz frequencies \cite{elshaer2016downlink, semiari2019integrated}. However, mm-wave transmissions are highly susceptible to blockages and penetration losses; therefore the mm-wave spectrum will complement the sub-6GHz spectrum in 5G networks \cite{ji2018ultra,andrews2014will,wang2015multi,polese2017improved}.

 In this article, we develop a framework to characterize the  meta  distributions  of  SIR/SNR  as  well  as  data  rate  in  the  coexisting  sub-6GHz and  mm-wave  cellular  network. We assume a two-tier network architecture as illustrated in Fig.~\ref{mmwave_network_arch}.  Tier 1 consists of macro base stations (MBSs) and tier 2 is composed of small base stations (SBSs).  A MBS communicates with SBSs on {\em backhaul links}  in the microwave spectrum.  SBSs communicate with devices on {\em access links} in the mm-wave spectrum.   This scenario supports {\em dual-hop} communications between MBSs and devices.  Devices can also communicate with MBSs via {\em direct links} in the microwave spectrum, as shown in Fig.~\ref{mmwave_network_arch}.

 Given the above {\em hybrid  spectrum network} architecture, it is crucial to develop new theoretic  frameworks to characterize the performance of such networks.  Within this context, we consider the use of {\em meta distributions} to study the performance of such hybrid spectrum networks.

%%{\color{black}Due to the congested and extremely limited sub-6GHz spectrum, the upcoming 5G networks will operate in both the conventional sub-6GHz spectrum and new mm-wave spectrum having wider bandwidths extending the data rates into the tens of Gbps \cite{ericsson_wp}.  Nevertheless,  mm-wave transmissions are highly susceptible to blockages and penetration losses, thus mm-wave spectrum and sub-6GHz spectrum will coexist in 5G networks \cite{polese2017improved,wang2015multi}.}

%%{\color{black}As such, to reap the benefits of 5G networks, the performance characterization of a  two-tier IoT  cellular  network  architecture with a tier of sub-6GHz  base  stations  (BSs)  and the other tier of mm-wave BSs is of immediate relevance. Also, different from the  coverage probability analysis, novel theoretic frameworks (such as meta distribution of the data rate or SIR) need be developed to analyze the reliability and latency of such networks.}

The meta distribution is first introduced by M. Haenggi~\cite{haenggi2016meta} to provide a  fine-grained {\em reliability and latency} analysis of 5G wireless networks with ultra-reliable and low latency communication requirements \cite{bennis2018ultrareliable,kalamkar2019per}. 
Meta distribution is defined as the distribution of the conditional success probability (CSP) of the transmission link (also termed as {\em link reliability}), conditioned on the locations of  the wireless transmitters. The meta distribution provides answers to questions such as ``{\em What fraction of {\color{black}devices} can achieve x\% transmission success probability?}'' whereas the conventional success probability answers questions such as {\em``What fraction of {\color{black}devices} experience transmission success?''} \cite{haenggi2016meta}. In addition to the standard coverage (or success) probability which is equivalent to the mean of CSP, the meta distribution can capture important network performance measures such as the mean of the local transmission delay, the variance of the local transmission delay (referred to as {\em network jitter}), and the variance of the CSP which depicts the variation of the {\color{black}devices}' performance from the mean coverage probability. Evidently, the standard coverage probability provides limited information about the performance of a typical wireless network \cite{salehi2017analysis,salehi2018meta,deng2019energy}. 
%To illustrate the significance of the meta distribution, assume that 50\% of the devices achieve 10\% reliability and the other 50\% achieve 99\% reliability. Then, the standard mean coverage probability is 54.5\%. On the other hand, if 100\% of the devices achieve 54.5\% reliability, the standard mean coverage probability is also 54.5\%. However, the two scenarios are very different in terms of  {\color{black}user} experience.   Meta distributions, on the other hand, can distinguish the above two scenarios.  %Furthermore, cellular operators are typically interested in the performance of the ``5\% device percentile'', which is the performance level that 95\% of the {\color{black}devices} achieve. The meta distribution reveals this information, while the standard coverage probability does not reveal any information about it. 

In this article, we develop a novel stochastic geometry framework based on meta distributions to estimate and analyze the communication latency and reliability of devices in a coexisting sub-6GHz and mm-wave  cellular network.
\begin{figure}[!h]
    \begin{center}
    \scalebox{0.3}[0.3]{\includegraphics{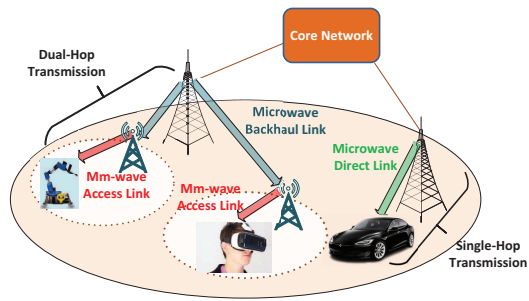}}
    %\includegraphics[width=\textwidth,keepaspectratio]{./figures/untitled.jpg}
    %\vspace{-.5cm}
    \end{center}
    \caption{\small Coexisting sub-6GHz and mm-wave cellular networks.}
    %The green circle represents MBS. The red circles and \color{red}e circles represent aggregators and devices, respectively
     \label{mmwave_network_arch}
    \end{figure}
\vspace{-1cm}
\subsection{Related Work}\label{related_work_mm}
A variety of research works studied the {\em coverage probability of  mm-wave only cellular networks} \cite{di2015stochastic, bai2015coverage, turgut2017coverage}. Di Renzo et al. \cite{di2015stochastic} proposed a general mathematical model to analyze multi-tier mm-wave cellular networks. Bai et al. \cite{bai2015coverage} derived the coverage and rate performance of mm-wave cellular networks. They used a distance dependent line-of-sight (LOS) probability function where the locations of the LOS and non-LOS (NLOS) BSs are modeled as two independent non-homogeneous Poisson point processes, to which different path loss models are applied. The authors assume independent Nakagami fading for each link. Different parameters of Nakagami fading are assumed for LOS and NLOS links. Turgut and Gursoy \cite{turgut2017coverage} investigated heterogeneous downlink mm-wave cellular
networks consisting of $K$ tiers of randomly located BSs where each tier operates in a mm-wave frequency band. They derived coverage probability for the entire network using tools from stochastic geometry. They used Nakagami fading to model small-scale fading. Deng et al. \cite{deng2018success} derived the success probability at the typical receiver in mm-wave device-to-device (D2D) networks. The authors considered Nakagami fading and incorporated directional beamforming.

Some recent studies  analyzed the {\em coverage or success probability} of coexisting $\mu$wave and mm-wave cellular networks. A hybrid cellular network  was considered by Singh et al. \cite{singh2015tractable} to estimate the uplink-downlink coverage and rate distribution of self-backhauled mm-wave networks. Elshaer et al. \cite{elshaer2016downlink} developed an analytical model to characterize decoupled uplink and downlink cell association strategies. The authors showed the superiority of this technique compared to the traditional coupled association in a network with traditional MBSs coexisting with denser mm-wave SBSs. Singh et al. \cite{singh2015tractable} and Elshaer et al. \cite{elshaer2016downlink} modeled the fading power as Rayleigh fading to enable better tractability. 

{\color{black} Compared to traditional coverage analysis conducted in \cite{turgut2017coverage,bai2015coverage,elshaer2016downlink}, Deng and Haenggi \cite{deng2017fine}  analyzed the meta distribution of the SIR in \emph{mm-wave only} {\em single-hop} D2D networks using the \emph{Poisson bipolar model} and simplified \emph{Rayleigh fading channels} for analytical tractability.}

\subsection{Contributions}

{\color{black}To the best of our knowledge, our work is the first to characterize the meta distributions of SIR/SNR and data rate for coexisting $\mu$wave and mm-wave networks. Different from previous research in \cite{turgut2017coverage,bai2015coverage,elshaer2016downlink,deng2017fine}, we develop a stochastic geometry framework that takes in consideration (i) coexistence of two different network tiers with completely different channel propagation, interference, and fading models, (ii) dual-hop transmissions enabled by two different spectrums, one in each network tier, and (iii)  Nakagami-m fading model with shape parameter $m$ for LOS mm-wave channels. Nakagami-m fading  is a generic and versatile distribution that includes Rayleigh distribution  (typically used
for non-LOS fading) as its special case when $m=1$ and can well approximate the Rician fading distribution for $1\leq m \leq \infty$ (typically used for 
LOS fading).}

%%The MBSs operate on the $\mu$wave spectrum while the SBSs communicate with {\color{black} devices} on the mm-wave spectrum. The SBSs are connected to MBSs via a $\mu$wave wireless backhaul (Fig.~\ref{mmwave_network_arch}). {\color{black}We term the coexisting $\mu$wave and mm-wave network  a {\em hybrid spectrum IoT network}}.
We assume a hybrid spectrum network architecture described above and illustrated in Fig. 1.
Since microwave transmissions are interference limited and mm-wave transmission are noise limited\footnote{Given highly directional beams and high sensitivity to blockage, recent studies showed that mm-wave networks can be considered as noise limited rather than interference limited \cite{ghosh2014millimeter,singh2015tractable,andrews2017modeling}.}, we study the meta distributions of the SIR and SNR in $\mu$wave and mm-wave channels, respectively.  We also characterize the meta distrubusiton of data rates.  Our contributions and methodology include the following:

\begin{itemize}
\item  Different from existing works, we characterize the CSP (which is equivalent to reliability) of a typical device  and its $b^{\mathrm{th}}$ moment when the {\color{black} device} either associates to (1)  $\mu$wave MBS for direct transmission or (2)  mm-wave SBS for dual-hop transmission (access and backhaul transmission). Using the novel moment expressions in the two scenarios, {\color{black}we derive a novel expression for the cumulative moment $M_{b,\text{T}}$ of the considered hybrid spectrum network.}
\item  Using the cumulative moment $M_{b,\text{T}}$, we characterize the exact and approximate meta distributions of the data rate and downlink SIR/SNR of a typical {\color{black}device}. 
 Since the expression of $M_{b, \text{T}}$ relies on a Binomial expansion of power $b$, the results for the meta-distribution requiring complex values of $b$ are obtained by applying Newton's Generalized Binomial Theorem.
 \item We characterize important network performance metrics such as coverage probability, mean local delay {\color{black}(which is equivalent to latency)}, and variance of the local delay (network jitter), using the derived cumulative moment $M_{b,\text{T}}$. For metrics with negative values of $b$, we apply the binomial theorem for negative integers.
 \item To model the LOS nature of mm-wave, we consider the versatile {\em Nakagami-m} fading channel model. To the best of our knowledge, the meta distribution for the Nakagami-{\em m} fading channel has not been investigated yet.
\item We demonstrate the application of this framework to other specialized network scenarios where (i)~SBSs are connected to MBSs via a mm-wave wireless backhaul and (ii)~a network where all transmissions are conducted in $\mu$wave spectrum.  Closed-form results are provided for special cases and asymptotic scenarios.
\end{itemize}
%where a {\color{black} device} can opportunistically associate to  with wireless backhauled SBSs.  The framework will enable {\color{black}IoT} cellular operators to analyze a wider range of system performance metrics including coverage probability, data rate, reliability, mean local delay, and network jitter. A summary of our contributions is listed herein:

%Our methodology of analysis for evaluating the meta distribution of SIR/SNR and rate is {\color{black}different from the existing literature since our model has} (i) two different network tiers  with completely different channel propagation, interference, and fading models, (ii) a dual-hop connectivity in the mm-wave tier as the SBSs are enabled with wireless backhauls, and (iii) we assume Nakagami-m fading for LOS mm-wave channels.}

We validate analytical results using Monte-Carlo simulations. Numerical results give valuable insights related to the  reliability,  mean local delay, variance of CSP, and standard success probability of a {\color{black} device}.  For example, the mean local delay increases with the increasing density of SBSs  in a $\mu$wave-only network; however, it stays constant in a hybrid spectrum  network. Moreover, the data rate reliability, i.e., the fraction of devices achieving a required data rate, increases as the number of antenna elements increases. We also note that as the number of antenna elements in  a hybrid spectrum  network increases, the reduction in the variance of reliability is noticeable, which shows the importance of analyzing the higher moments of the CSP using the meta distribution. These insights would help 5G cellular network operators to find the most efficient operating antenna configurations for ultra-reliable and low latency  applications.

\subsection{Outline of the Article}
The remainder of the article is organized as follows. In Section \ref{model_mm}, we describe the system model and assumptions. In Section \ref{meta_dist_analysis}, we provide mathematical preliminaries of the meta distribution. In Section~\ref{total_moment_assoc}, we characterize the association probabilities of a typical {\color{black} device} and formulate the meta distribution of the  SIR{\color{black}/SNR} of a {\color{black}device} in the hybrid spectrum 5G cellular networks. In Section~\ref{v_sec}, we characterize the CSP and its $b^{\mathrm{th}}$ moment for direct, access, and backhaul transmissions. Finally, we derive the exact and approximate meta distributions of the SIR{\color{black}/SNR} and data rate in a hybrid spectrum network as well as $\mu$wave-only network in Section \ref{mu-wave-meta}.  Finally, Section \ref{results_mm}  presents  numerical results and Section \ref{conc_future_mm} concludes the article.

%%%%%%%%%%%%%%%%%%%%%%%%%%%%%%%%%%%%%%%%%%%%%%%%%%%%%%%%%%%%%%%%%%%%%%%%%%%%%%%%%%%%%%%%%%%%%%%%%%%%%%%%%%%%
\section{System Model and Assumptions}\label{model_mm}
In this section, we describe the network deployment model (Section \ref{network_model_mmwave}), antenna model (Section \ref{beamforming}), channel model (Section \ref{channel_model}), {\color{black}device} association criteria (Section \ref{biased_association}), and SNR/SIR models for access and backhaul transmissions (Section \ref{snr-sir}).

\subsection{Network Deployment and Spectrum Allocation Model}\label{network_model_mmwave}
    We assume a two-tier cellular network architecture as shown in Fig.~\ref{mmwave_network_arch} in which the locations of the MBSs and SBSs are modeled as a two-dimensional (2D) homogeneous Poisson point process (PPP) $\mathbf{\Phi}_k=\{\mathbf{y}_{k,1}, \mathbf{y}_{k,2},...\}$ of density $\lambda_{k}$,  where $\mathbf{y}_{k,i}$ is the location of  $i^{\mathrm{th}}$ MBS (when $k=1$) or the $i^{\mathrm{th}}$ SBS (when $k=2$). Let the MBS tier be tier~1 ($k=1$) and the SBSs constitute tier~2 ($k=2$). Let $\mathcal{D}$ denotes the set of {\color{black} devices}. The locations of {\color{black} devices} in the network are modeled as independent homogeneous PPP  $\mathbf{\Phi}_\mathcal{D}=\{\mathbf{x}_{1}, \mathbf{x}_{2},....\}$ with density $\lambda_{\mathcal{D}}$, where $\mathbf{x}_{i}$ is the location of the $i^{\mathrm{th}}$
    {\color{black}device}. We assume that $\lambda_{\mathcal{D}}\gg\lambda_{2}>\lambda_{1}$ as in
    \cite{Hazem2015,ibrahim2018data,ibrahim2016mobility}. We consider a typical outdoor {\color{black}device} which is located at the origin and is denoted by $0$ and its tagged BS is denoted by $\mathbf{y}_{k,0}$, i.e., tagged MBS (when $k=1$) or tagged SBS (when $k=2$). All BSs in the $k^{\mathrm{th}}$ tier transmit with the same transmit power $P_{k}$ in the downlink. A list of the key mathematical notations  is given in Table~\ref{Notation_Summary_mmwave}.
    
     We assume that a portion $\eta W_1$ of the frequency band $W_1$ is reserved for the access transmission and the rest $(1-\eta) W_1$ is reserved for the backhaul transmission, where $W_1$, and $W_2$ denote the total available $\mu$wave spectrum and mm-wave spectrum, respectively, and $0 \leq \eta \leq 1$. Determining the optimal spectrum allocation ratio $\eta$ will be studied in our future work. 
     %{\color{black}It will not affect the meta distribution of coverage but it will affect the meta distribution of rate discussed in Section \ref{data_rate_meat}.}
\footnotesize
\begin{table*}[!h]
\centering
\small
\caption{{\color{black} Mathematical Notations} }
\resizebox{\textwidth}{!}{\begin{tabular}{|c|c|c|c|}
\hline
  \textbf{Notation}&\textbf{Description} & \textbf{Notation}&\textbf{Description}\\ \hline
  \footnotesize$\mathbf{\Phi}_k$;$\mathbf{\Phi}_\mathcal{D}$&\footnotesize PPP of BSs of $k^{\mathrm{th}}$ tier; PPP of {\color{black}devices} & \footnotesize $\lambda_k$;$\lambda_{\mathcal{D}}$&\footnotesize Density of BSs of $k^{\mathrm{th}}$ tier; density of {\color{black}devices}\\ \hline
  \footnotesize $P_k$&\footnotesize\shortstack{Transmit power of BSs in $k^{\mathrm{th}}$ tier} & $B_k$&\shortstack{Association bias for BSs of $k^{\mathrm{th}}$ tier}\\ \hline
  \footnotesize$\alpha_1,\alpha_{2,L},\alpha_{2,N}$&\footnotesize\shortstack{Path loss exponent of MBS tier;\\ LOS SBS; NLOS SBS} & \footnotesize$G^{o}_{1}$&\shortstack{\footnotesize omnidirectional antenna gain of $\mu$wave MBSs }\\ \hline
  \footnotesize$G_{2}^{max}$;$G_{2}^{min}$;$\theta_{a}$&\shortstack{\footnotesize Main lobe gain; side lobe gain; and \\\footnotesize 3 dB beamwidth for mm-wave SBS} &  \footnotesize$h_{l}$&\shortstack{\footnotesize Gamma fading channel gain for mm-wave SBSs}\\ \hline
  \footnotesize$g$&\shortstack{\footnotesize Rayleigh fading channel gain} &  \footnotesize$m_{l}$&\shortstack{\footnotesize Nakagami-m fading parameter where $l\in\{L,N\}$  \\\footnotesize denotes LOS and NLOS transmission links}\\ \hline
  $p_{L}$;$p_{N}$&\shortstack{\footnotesize Mm-wave blockage LOS probability; NLOS probability} &
  $\theta$&\shortstack{\footnotesize Predefined SIR/SNR threshold} \\ \hline
  \footnotesize$\bar{F}_{P_{s}}(x)$ &\shortstack{\footnotesize Meta distribution of SIR/SNR}  &  \footnotesize$P_{s}(\theta)$ &\shortstack{\footnotesize Conditional success probability (CSP)} \\ \hline
  \footnotesize$M_{b}(\theta)$& \footnotesize The $b^{th}$ moment of $P_{s}(\theta)$ &  \footnotesize$\mathcal{A}_{2}$;$\mathcal{A}_{2,L}$;$\mathcal{A}_{2,N}$& \shortstack{\footnotesize Association Probability with $\mu$wave MBS;\\\footnotesize LOS mm-wave SBS; NLOS mm-wave SBS} \\ \hline
\end{tabular}}
\label{Notation_Summary_mmwave}
\end{table*}
\normalsize

\subsection{Antenna Model}\label{beamforming}
We assume that all MBSs are equipped with omnidirectional antennas with gain denoted by $G^{\mathrm{o}}_{1}$ dB.
We consider SBSs and {\color{black}devices} are equipped with directional antennas with  sectorized gain patterns as in  \cite{andrews2017modeling,di2015stochastic,deng2017fine} to approximate the actual antenna pattern. 
%To compensate for the large path loss in the mm-wave spectrum, directional beamforming is necessary for communication \cite{lee2016randomly}. 
The sectorized gain pattern is given by:
\begin{equation}\label{sets}
  \hspace{-.5 cm} G_{a}(\theta) = \begin{cases}
    G_{a}^\mathrm{max}& \mbox{if } \left|\theta\right| \leq \frac{\theta_{a}}{2} \\
    G_{a}^\mathrm{min}&\mbox{otherwise }
  \end{cases},
\end{equation}
where subscript $a\in\{2,\mathcal{D}\}$ denotes for SBSs and {\color{black}devices}, respectively.  Considering a $\sqrt{\mathcal{N}}\times\sqrt{\mathcal{N}}$ uniform planar square antenna array with $\mathcal{N}$ elements, the antenna parameters of a uniform planar square antenna array can be given as in \cite{deng2017fine}, i.e., $G_{a}^\mathrm{max}=\mathcal{N}$ is the main lobe antenna gain, $G_{a}^\mathrm{min}=1/\sin^{2}\left(\frac{3\pi}{2\sqrt{\mathcal{N}}}\right)$ is the side lobe antenna gain, $\theta\in [-\pi,\pi)$ is the angle of the boresight direction, and $\theta_{a}=\frac{\sqrt{3}}{\sqrt{\mathcal{N}}}$ is the main lobe beam width. A perfect beam alignment is assumed between a {\color{black}device} and its serving SBS \cite{elshaer2016downlink} \cite{bai2015coverage}. The antenna beams of the desired access links are assumed to be perfectly aligned, i.e., the direction of arrival (DoA) between the transmitter and receiver is known a priori at the BS and the effective gain on the intended access link can thus be denoted as $G_{2}^\mathrm{max} G_{\mathcal{D}}^\mathrm{max}$. This can be done by assuming that the serving mm-wave SBS and {\color{black}device} can adjust their antenna steering orientation using the estimated angles of arrivals. The analysis of the alignment errors on the desired link is beyond the scope of this work.  

%\begin{table}[ht]
%\caption{antenna parameters of a uniform planar
%square antenna array [16].}   % title of Table
%\centering % used for centering table
%\resizebox{0.32 \textwidth}{!}{\begin{tabular}{c c c} % centered columns (4 columns)
%\hline\hline %inserts double horizontal lines
%Parameter&Description&Value \\ [0.5ex] % inserts table
%%heading
%\hline % inserts single horizontal line
%0.1&-5.5795&0.6\\ % inserting body of the table
%0.5&-2.4679&0.3\\
%0.75&-2.3266&0.3\\
%1&-2.4679&0.4\\ [1ex] % [1ex] adds vertical space
%\hline %inserts single line
%\end{tabular}}
%\label{orth_aggreg_PCFss} % is used to refer this table in the text
%\end{table}

\subsection{Channel Model}\label{channel_model}
\subsubsection{Path-Loss Model} The signal power decay is modeled as $ L(r)=r^{\alpha}$, where  $L(r)$ is the path loss for a typical receiver located at a distance $r$ from the transmitter and $\alpha$ is the path loss exponent (PLE).  Let $L_{1}(r)=\|r_{1,\mathcal{D}}\|^{\alpha_{1}}$ denotes
the path loss of a typical {\color{black}device} associated with the MBS tier, where $\alpha_1$ is the PLE. Similarly, $L_{2}(r)=\|r_{2,\mathcal{D}}\|^{\alpha_{2,l}}$  denotes
the path loss of a typical {\color{black}device} associated with the SBS tier where $\alpha_{2,l}=\alpha_{2,L}$ is the PLE in the case of LOS and $\alpha_{2,l}=\alpha_{2,N}$ is the PLE in the case of NLOS.  It has been shown that mm-wave LOS and NLOS conditions have markedly different PLEs \cite{maccartney2013path}. Also, we consider the near-field path loss factor $\zeta=(\frac{\text{carrier wavelength}}{4\pi})^{2}$ at 1 m \cite{elshaer2016downlink}, i.e., different path loss for different frequencies at the reference distance.
%%%%%%%%%%%%%%%
\subsubsection{Fading Model}
 For outdoor mm-wave channels, we consider a versatile Nakagami-m  fading channel model due to its analytical tractability and following the previous line of research studies \cite{bai2015coverage,turgut2017coverage,khan2016millimeter,yi2017modeling,deng2018success}. Nakagami-{\em m} fading is a general and tractable model to characterize mm-wave channels. Also, in several scenarios, Nakagami-m can approximate the Rician fading  which is commonly used to model the LOS transmissions but not tractable for meta distribution modeling \cite{nakagami1960m,simon2005digital}.
The fading parameter $m_{l}\in [1,2,...,\infty)$ where $l\in\{L,N\}$ denotes LOS and NLOS transmission links, respectively, and the mean fading power is denoted by $\Omega_{l}$.   The fading channel power $h_{l}$ follows a gamma distribution given as $f_{h_{l}}(x)=\frac{m_{l}^{m_{l}}x^{m_{l}-1}}{\Omega_{l}^{m_{l}}\Gamma(m_{l})}\exp(\frac{-m_{l}x}{\Omega_{l}})$, $x>0$, where $\Gamma(.)$ is the Gamma function, $m_{l}$ is the shape (or fading) parameter, and $\frac{m_{l}}{\Omega_{l}}$ is the scale parameter. That is, we consider $h_{l}\sim\Gamma(m_{L},1/m_{L})$ for the LOS links and $h_{l}\sim\Gamma(m_{N},1/m_{N})$ for the NLOS links. Rayleigh fading is a special case of Nakagami-{\em m} for $m_L=m_N=1$. Due to the NLOS nature of $\mu$wave channels, we assume Rayleigh fading with power normalization, i.e., the channel gain $g(\mathbf{x}, \mathbf{y})\sim \exp(1)$, is independently  distributed with the unit mean.

%%%%%%%%%%%%%%
\subsubsection{Blockage Model for Mm-wave Access Links}\label{blockage_parag}
For mm-wave channels, LOS transmissions are vulnerable to significant penetration losses \cite{maccartney2013path}; thus LOS transmissions can be blocked with a certain probability. Following \cite{bai2015coverage,venugopal2016device,khan2016millimeter,bai2014analysis}, we consider the actual LOS region of a {\color{black}device} as a fixed LOS ball referred to as "equivalent LOS ball". For the sake of mathematical tractability, we consider a distance dependent blockage probability $p(r)$ that a mm-wave link of length $r$ observes, i.e., the LOS probability $p_{L}(r)$ if the mm-wave desired link length is less than $d$ and $p_{N}(r)$ otherwise. That is, SBSs within a LOS ball of radius $d$ are marked LOS with probability $p_{L}(r)$, while the SBSs outside that LOS ball are marked as NLOS with probability $p_{N}(r)$. Note that we will drop the notation $(r)$ in both $p_{L}(r)$ and $p_{N}(r)$ from this point onwards and we will use only $p_{L}$ and $p_{N}$, respectively.

%The parameter $\ddot{L}$ is the blockage density parameter, which is obtained from the size and density of blocking obstacles which can be interpreted as the average fraction of LOS area in a circular ball of radius $D$ around the point under consideration. The parameters $(\ddot{L}, D)$ are geography and deployment dependent (low for dense urban, high for semi-urban). Therefore, in our model, we consider the Manhattan scenario from \cite{singh2015tractable} where $D=200$ meter $\ddot{L}=0.11$. Results for other values of $(\ddot{L}, D)$ are discussed too in Section  to study their impact on total network meta distribution.

\subsection{Association Mechanism}\label{biased_association}
%Let an user be located at $\mathbf{x}$ and associated with a tagged node of tier $k$ at location $\mathbf{y}$. The received power at the user can then be given as $P_{k}\|\mathbf{y}-\mathbf{x}\|^{-\alpha_{k}}h(\mathbf{x}, \mathbf{y})$, based on the associated tier.
 Each {\color{black}device} associates with either a MBS or a SBS depending on the maximum biased received power in the downlink. The  association criterion at the typical {\color{black}device} can be written mathematically as follows:
 \small
\begin{equation}\label{biasing}
    P_{k}B_{k}G_{k}\zeta_{k}L_{k}(r)^{-1}\ge P_{j}B_{j}G_{j}\zeta_{j}L_{\mathrm{min},j}(r)^{-1}, \forall  j\in\{1,2\}, j\neq k
    \end{equation}
    \normalsize
    \noindent where $P_{(\cdot)}$, $B_{(\cdot)}$, $G_{(\cdot)}$, and $\zeta_{(\cdot)}$ denote the transmission power, biasing factor, effective antenna gain, and near-field path loss at 1 m of the intended link, respectively, in the corresponding tier (which is determined by the index in the subscript).  Let $L_{\mathrm{min},j}(r)^{-1}$ be the minimum path loss of a typical {\color{black}device} from a BS in the $j^{\mathrm{th}}$ tier. When a {\color{black}device} associates with a mm-wave SBS in tier-2, i.e., $k=2$, the antenna gain of the intended link is  $G_{2}=G_{2}^\mathrm{max} G_{\mathcal{D}}^\mathrm{max}$, otherwise $G_{1}=G^{\mathrm{o}}_{1} G_{\mathcal{D}}$, where $G^{\mathrm{o}}_{1}$ is defined as the omnidirectional antenna gain of MBSs and $ G_{\mathcal{D}}$ is the {\color{black}device} antenna gain while operating in $\mu$wave spectrum. On the other hand, the SBS associates with a MBS offering the maximum received power in the downlink.

 \subsection{SNR/SIR Models for Access and Backhaul Transmissions}\label{snr-sir}
The {\color{black}device}  associates to either a MBS for direct transmission or a SBS for dual-hop transmission. The first link (backhaul link) transmissions occur on the $\mu$wave spectrum between MBSs and SBSs and the second link (access link) transmissions take place in the mm-wave spectrum between SBSs and {\color{black}devices}. Let $\theta_{2}$ denotes the predefined SIR threshold for SBSs in the backhaul link and $\theta_{\mathcal{D}}$ denotes the predefined SIR/SNR threshold for {\color{black}devices}. Throughout the paper, we use subscripts ``$1,2$'', ``$2,\mathcal{D}$'', ``$1,\mathcal{D}$'',  ``$\mathcal{D}$'', ``$\mathrm{BH}$'' to denote backhaul link, access  link, direct  link, {\color{black}device}, and backhaul,  respectively.
\subsubsection{Backhaul Transmission}
The $\mathrm{SIR}$ of a typical SBS associated with a MBS can be modeled as:
       \begin{equation} \label{SIR-backhaul-O}
       \mathrm{SIR}_{1,2} = \frac {P_{1}r_{1,2}^{-\alpha_{1}}g(0, \mathbf{y}_{1,0})}{\mathcal{I}_{1,2}},
       \end{equation}
       where $\mathcal{I}_{1,2}$ denotes the backhaul interference received at a SBS from MBSs that are scheduled to transmit on the same resource block excluding the tagged MBS. Then,
      $\mathcal{I}_{1,2}=P_{1}\sum_{i:\mathbf{y}_{1,i}\in\mathbf{\Phi}_1\backslash \{\mathbf{y}_{1,0}\}}\|\mathbf{y}_{1,i}\|^{-\alpha_{1}}g(0,\mathbf{y}_{1,i}).$
\subsubsection{Direct Transmission} The $\mathrm{SIR}$ of a typical {\color{black}device} associated directly with a MBS is modeled as:
       \begin{equation} \label{SIR-direct-link}
       \mathrm{SIR}_{1,\mathcal{D}} = \frac{P_{1}r_{1,\mathcal{D}}^{-\alpha_{1}}g(0, \mathbf{y}_{1,0})}{\mathcal{I}_{1,\mathcal{D}}},
       \end{equation}
       where $\mathcal{I}_{1,\mathcal{D}}$ denotes the interference received at a typical {\color{black}device} from  MBSs  excluding the tagged MBS. Then $\mathcal{I}_{1,\mathcal{D}}$ can be  calculated as:
       $
        \mathcal{I}_{1,\mathcal{D}}=P_{1}\sum_{i:\mathbf{y}_{1,i}\in\mathbf{\Phi}_1\backslash \{\mathbf{y}_{1,0}\}}\|\mathbf{y}_{1,i}\|^{-\alpha_1}g(0,\mathbf{y}_{1,i}).
        $
\subsubsection{Access Transmission}
The SNR of a typical {\color{black}device}  associated with a mm-wave SBS  is modeled as:
         \begin{equation} \label{snr-access-O-mmwave}
       \mathrm{SNR}_{2,\mathcal{D}} = \frac {P_{2}G_{2}\zeta_{2}\|r_{2,\mathcal{D}}\|^{-\alpha_{2,l}}h_{l}(0, \mathbf{y}_{2,0})}{\sigma_{2}^{2}},
       \end{equation}
where $\zeta_{2}$ is the near-field path loss at 1 m for mm-wave channels,  and $\sigma_{2}^{2}$ is the variance of the additive white Gaussian noise at the {\color{black}device} receiver. Given highly directional beams and high sensitivity to blockage, recent studies showed that mm-wave networks are typically noise limited \cite{ghosh2014millimeter,singh2015tractable,andrews2017modeling}.
%%%%%%%%%%%%%%%%%%%%%%%%%%%%%%%%%%%%%%%%%%%%%%%%%%%%%%%%%%%%%%%%%%%%%%%%%%%%%%%%%%%%%%%%%%%%%%%%%%%%%%%%%%%%%%%%%%%%%%%%%%%%
\section{The Meta Distribution: Mathematical Preliminaries}\label{meta_dist_analysis}
In this section, we define the meta distribution  of the SIR of a typical {\color{black}device} and highlight exact and  approximate methods to evaluate the meta distribution. 

%%%%%%%%%%%%%%%%%%%%%%%%%%%%%%%%%%%%%%%%%%%%%%%%%%%%%%%%%%%
\begin{Definition}[Meta Distribution of the SIR and CSP]
The meta distribution $\bar{F}_{P_{s}}(x)$ is the complementary cumulative distribution function (CCDF) of the CSP (or reliability) $P_{s}(\theta)$ and given by \cite{haenggi2016meta}:
\begin{equation}\label{meta-1}
\bar{F}_{P_{s}}(x)\overset{\Delta}{=}\mathbb{P}(P_{s}(\theta)>x), \quad\quad x\in [0,1],
\end{equation}
where, conditioned on the locations of the transmitters
and that the desired transmitter is active, the CSP $P_{s}(\theta)$ of a typical {\color{black}device} \cite{haenggi2016meta} can be given as
$
P_{s}(\theta)\overset{\Delta}{=}\mathbb{P}(\mathrm{SIR}>\theta|\mathbf{\Phi},\text{tx})
$
where $\theta$ is the desired $\mathrm{SIR}$. 
\end{Definition}

%%%%%%%%%%%%%%%%%%%%%%%%%%%%%%%%%%%%%%%%%%%%%%%%%%%%%%%%%%%
Physically, the meta distribution provides the fraction of the active links whose CSP (or reliability) is greater than the reliability threshold $x$.
Given $M_{b}(\theta)$ denotes the $b^{\mathrm{th}}$ moment of $P_{s}(\theta)$, i.e., $M_{b}(\theta)\overset{\Delta}{=}\mathbb{E}^{0}(P_{s}(\theta)^{b})$, $b\in\mathbb{C}$,
the exact meta distribution can be given using the Gil-Pelaez theorem \cite{gil1951note}  as {\color{black}\cite{haenggi2016meta}}:
\begin{equation}\label{Gil-pelaez-2hops}
\bar{F}_{P_{s}}(x)=\frac{1}{2}+\frac{1}{\pi}\int_{0}^{\infty}\frac{\Im\left(e^{-jt\log{x}}M_{jt}(\theta)\right)}{t}\text{d}t,
\end{equation}
where $\Im(w)$ is imaginary part of $w\in \mathbb{C}$ and $M_{jt}(\theta)$ denotes  the imaginary moments of $P_{s}(\theta)$, i.e., , $j\overset{\Delta}{=}\sqrt{-1}$. Using moment matching techniques and taking $\beta\overset{\Delta}{=}\frac{(M_{1}(\theta)-M_{2}(\theta))(1-M_{1}(\theta))}{M_{2}(\theta)-M_{1}(\theta)^{2}}$, the {meta distribution of the CSP} can be approximated using the Beta distribution as follows:
\begin{equation}\label{beta-match}
\bar{F}_{P_{s}}(x)\approx 1-I_{x}\left(\frac{\beta M_{1}(\theta)}{1-M_{1}(\theta)},\beta\right), \quad\quad x\in [0,1],
\end{equation}
where   $M_{1}(\theta)$ and $M_{2}(\theta)$ are the first and the second moments, respectively; $I_{x}(a,b)$ is the regularized incomplete Beta function
$
I_{x}(a,b)\overset{\Delta}{=}\frac{\int_{0}^{x}t^{a-1}(1-t)^{b-1}d\text{t}}{B(a,b)}
$
and $B(a,b)$ is the Beta function. 

%%%%%%%%%%%%%%%%%%%%%%%%%%%%%%%%%%%%%%%%%%%%%%%%%%%%%%%%%%%%%%%%%%%%%%%%%%%%%%%%
\section{The Meta Distribution of the SIR/SNR in Hybrid Spectrum Networks}\label{total_moment_assoc}
To characterize the meta distribution of the SIR/SNR of a typical {\color{black}device} that can associate with either a $\mu$wave MBS  with probability $\mathcal{A}_{1}$ or with a wireless backhauled mm-wave SBS  with probability $\mathcal{A}_{2}$, the methodology of analysis is listed as follows:
\begin{enumerate}
    \item Derive the probability of a typical  {\color{black}device} associating with $\mu$wave MBSs $\mathcal{A}_{1}$, LOS mm-wave SBSs $\mathcal{A}_{2,{\rm L}}$, and NLOS mm-wave SBSs $\mathcal{A}_{2,{\rm N}}$ where $\mathcal{A}_2=\mathcal{A}_{2,{\rm L}}+\mathcal{A}_{2,{\rm N}}$ ({\bf Section~\ref{A_222}}).
    \item Formulate the meta distribution of the SIR/SNR of a {\color{black}device} in the  hybrid network ($\bar{F}^b_{P_{s,\text{T}}}(x)$) considering the direct link and dual-hop link with wireless backhaul transmission ({\bf Section~\ref{meta_dual}}). 
    \item Formulate the CSP (${P_{s,\text{T}}}(\theta)$) and its $b^{\mathrm{th}}$ moment  $(M_{b,\text{T}})$ ({\bf Section~\ref{meta_dual}}).
    \item Derive the CSP at backhaul link $P_{s,\mathrm{BH}}(\theta_{2})$, CSP at access link $P_{s,2}(\theta_{\mathcal{D}})$, and CSP at direct link $P_{s,1}(\theta_{\mathcal{D}})$. Derive the $b^{\mathrm{th}}$ moments of CSPs, i.e., $M_{b,\mathrm{BH}}(\theta_{2})$, $M_{b,2}(\theta_{\mathcal{D}})$, and $M_{b,1}(\theta_{\mathcal{D}})$ for  backhaul link, access link, and direct link transmissions, respectively ({\bf Section~\ref{v_sec}}). 
    \item Obtain the meta distributions of SIR/SNR and data rate in hybrid spectrum  network using Gil-Pelaez inversion and the Beta approximation ({\bf Section~\ref{vi_sec}}).
\end{enumerate}

%%%%%%%%%%%%%%%%%%%%%%%%%%%%%%%%%%%%%%%%%%%%%%%%%%%%%%%%%%%%%%%%%%%%%%%%%%%%%%%%%%%%%%%%%%%%%%%%%%%%%%%%%%%%%%%%%%%%%%%%%%%%%%
\subsection{Association Probabilities in Hybrid Spectrum Networks}\label{A_222}
In this subsection, we characterize the probabilities with which a typical {\color{black}device} associates with $\mu$wave MBSs ($\mathcal{A}_{1}$) or mm-wave SBSs ($\mathcal{A}_{2}$). The results are presented in the following.
%%%%%%%%%%%%%%%lemma 1%%%%%%%%%%%%%%%%
\begin{lemma}[The Probability of Associating with mm-wave SBSs]\label{1}
	The  probability of a typical {\color{black}device} to associate with a mm-wave SBS, using the association scheme  in Eq. \eqref{biasing}, can be expressed as:
	\small
	\begin{align}
	&\mathcal{A}_{2}=1-\frac{2\pi\lambda_{1}}{\hat{a}\alpha_{1}}\Bigg(\int_{0}^{d^{\alpha_{2,L}}}H(l_1)e^{-\pi\lambda_{2}p_{L}l_{1}^{\frac{2}{\alpha_{2,L}}}}d\text{l}_{1}+\int_{d^{\alpha_{2,L}}}^{d^{\alpha_{2,N}}}H(l_1)e^{-\pi\lambda_{2}p_{L}d^{2}}d\text{l}_{1}+
	%\notag\\
	%&\quad\quad\quad\quad\quad\quad\quad\quad\quad\quad\quad\quad\quad\quad\quad\quad\quad\quad\quad\quad\quad\quad\quad\quad
	\int_{d^{\alpha_{2,N}}}^{\infty}H(l_1)e^{-\pi\lambda_{2}\bigg[(p_{L}-p_N)d^{2}+p_{N}l_{1}^{\frac{2}{\alpha_{2,N}}}\bigg]}
	d\text{l}_{1}\Bigg),
	\end{align}
	\normalsize
	where $\hat{a}\overset{\Delta}{=}\frac{P_{2}B_{2}G_{2}\zeta_{2}}{P_{1}B_{1}G_{1}\zeta_{1}}$ and $H(l_1)\overset{\Delta}{=}\left(\frac{l_1}{\hat{a}}\right)^{\frac{2}{\alpha_{1}}-1}\exp\left(-\pi\lambda_{1}\left(\frac{l_1}{\hat{a}}\right)^{\frac{2}{\alpha_{1}}}\right)$. Subsequently, the probability of a {\color{black}device} to associate with a $\mu$wave MBS can be given as $\mathcal{A}_{1}=1-\mathcal{A}_{2}$. The conditional association probability for a typical {\color{black}device} to associate with SBS is as follows:
	\begin{align}
	\bar{\mathcal{A}}_{2}(l_1)&=1-\frac{2\pi\lambda_{1}}{\hat{a}\alpha_{1}}\Bigg(H(l_1)e^{-\pi\lambda_{2}p_{L}l_{1}^{\frac{2}{\alpha_{2,L}}}}+H(l_1)e^{-\pi\lambda_{2}p_{L}d^{2}}+
	%\notag\\
	%&\quad\quad\quad\quad\quad\quad\quad\quad\quad\quad\quad\quad\quad\quad\quad\quad\quad\quad\quad\quad\quad\quad\quad\quad
	H(l_1)e^{-\pi\lambda_{2}\left[(p_{L}-p_N)d^{2}+p_{N}l_{1}^{\frac{2}{\alpha_{2,N}}}\right]}
	\Bigg),
	\end{align}
	\normalsize
	subsequently, $\bar{\mathcal{A}}_{1}(l_1)=1-\bar{\mathcal{A}}_{2}(l_1)$.
\end{lemma}
%\vspace{-0.8cm}
\begin{proof}
Using the approach in \cite{elshaer2016downlink}, we derive Lemma 1 in {\bf Appendix A} of
our technical report \cite{dblp}.
\end{proof}
A closed-form expression of $\mathcal{A}_1$ can be derived for a case of practical interest as follows. 
\begin{Corollary}
When $\alpha_{1}=4$, $\alpha_{2,L}=2$, and $\alpha_{2,N}=4$, then $\mathcal{A}_1$ can be given in closed-form as follows: %(\ref{close_form}),
%\begin{figure*}[ht!]
\begin{align}\label{close_form} 
\mathcal{A}_1=\frac{e^C
(\Phi[\sqrt{C}+\sqrt{\pi\lambda_2 p_L d^2}]-\Phi[\sqrt{C}])
}
{\sqrt{p_L \lambda_2/\hat{a}}}
+
\frac{e^{-d^2 \pi p_L \lambda_2}(e^{-\pi \lambda_1 \sqrt{d^2/\hat{a}}}-e^{-\pi \lambda_1 \sqrt{d^4/\hat{a}}})}{\pi \lambda_1/2 \hat{a}}
+
\frac{e^{d^2 \pi (p_N-p_L) \lambda_2-C_1 \sqrt{d^4/\hat{a}}})}{ C_1/2\hat{a}},
	\end{align}
	%\hrule
%\end{figure*}
	where $\Phi(\cdot)$ is the error function,  $C=\frac{\pi \lambda^2_1}{4 \hat{a} p_L \lambda_2}$ and $C_1=\pi (\lambda_1+\sqrt{\hat{a}}p_N\lambda_2)$ and  $\mathcal{A}_{2}=1-\mathcal{A}_{1}$.
\end{Corollary}
It can be seen from {\bf Corollary~1} that when the number of antenna elements $\mathcal{N}$ goes to infinity, i.e., $G_2 \rightarrow \infty$, $\hat{a} \rightarrow \infty$,  then $\mathcal{A}_1$ can be simplified as 
$
\mathcal{A}_1=\frac{
\Phi[\sqrt{\pi\lambda_2 p_L d^2}]
}
{\sqrt{p_L \lambda_2/\hat{a}}}
+
\frac{e^{d^2 \pi (p_N-p_L) \lambda_2}}{ C_1/2\hat{a}},
$ which shows that association probability to MBS will be very small. Similar insights can be extracted for other parameters.
 
 In order to derive the $b^{\mathrm{th}}$ moment of CSP $P_{s,2}(\theta_{\mathcal{D}})$ on an access link when a {\color{black}device} associates with a SBS (the CSP will be discussed later in Lemma \ref{3}), we have to derive the probability of a {\color{black}device} to associate  with LOS SBS $\mathcal{A}_{2,L}$ and NLOS SBS $\mathcal{A}_{2,N}$ which are defined follows.
\begin{lemma}[The Probability of Associating with LOS and NLOS mm-wave SBSs]\label{2}
When a typical {\color{black}device} associates with the mm-wave SBS tier, this typical {\color{black}device} has two possibilities to connect to (a) a LOS mm-wave SBS with association probability
	$\mathcal{A}_{2,L}$ and (b) a NLOS mm-wave SBS with association probability
	$\mathcal{A}_{2,N}$  which are characterized, respectively, as follows:
	\small
	\begin{align}\label{ass_SBS_LOS}
	\mathcal{A}_{2,L}&=\int_{0}^{d^{\alpha_{2,L}}} \bar{\mathcal{A}}_{2,L}(l_{2,L}) d\text{l}_{2,L},\quad\mathcal{A}_{2,N}=\int_{d^{\alpha_{2,N}}}^{\infty} \bar{\mathcal{A}}_{2,N}(l_{2,N}) d\text{l}_{2,N},
	\end{align}
	\normalsize
	where $\bar{\mathcal{A}}_{2,L}(l_{2,L})$ and $\bar{\mathcal{A}}_{2,N}(l_{2,N})$  are the conditional probabilities with which a typical {\color{black}device} may associate to the LOS and NLOS mm-wave SBSs, respectively, and are defined as follows:
	\small
	\begin{align*} \bar{\mathcal{A}}_{2,L}(l_{2,L})\overset{\Delta}{=}&\frac{2\pi\lambda_{2}p_{L}}{\alpha_{2,L}}l_{2,L}^{\frac{2}{\alpha_{2,L}}-1}e^{\left(-\pi\lambda_{1}\left(\bar{a}l_{2,L}\right)^{\frac{2}{\alpha_{1}}}-\pi\lambda_{2}p_{L}l_{2,L}^\frac{2}{\alpha_{2,L}}\right)},\\
	\bar{\mathcal{A}}_{2,N}(l_{2,N})\overset{\Delta}{=}&
	\frac{2\pi\lambda_{2}p_{N}}{\alpha_{2,N}}l_{\text{2,N}}^{\frac{2}{\alpha_{2,N}}-1}\exp\left(-\pi\lambda_{1}\left(\bar{a}l_{2,N}\right)^{\frac{2}{\alpha_{1}}}-\pi\lambda_{2}\big[p_{L}d^{2}+p_{N}(l_{2,N}^\frac{2}{\alpha_{2,N}}-d^{2})\big]\right)d\text{l}_{2,N},
	\end{align*}
	\normalsize
	where $\bar{a}\overset{\Delta}{=}\frac{P_{1}B_{1}G_{1}\zeta_{1}}{P_{2}B_{2}G_{2}\zeta_{2}}$, $\bar{\mathcal{A}}_{2}(l_{2})=\bar{\mathcal{A}}_{2,L}(l_{2,N})+\bar{\mathcal{A}}_{2,N}(l_{2,N})$ and $\mathcal{A}_{2}=\mathcal{A}_{2,L}+\mathcal{A}_{2,N}$.
\end{lemma}
%\vspace{-0.8cm}
\begin{proof}
	Using the approach in \cite{turgut2017coverage}, we derive Lemma 2 in {\bf Appendix B} of
our technical report \cite{dblp}.
\end{proof}
A case of interest is when the number of antenna elements at mm-wave SBSs increases asymptotically. In such a case, the LOS and NLOS association probabilities can be simplified as follows: 
\begin{Corollary} When the number of antenna elements at mm-wave SBSs increases, i.e., $\mathcal{N}\rightarrow \infty$, $\alpha_{1}=4$, $\alpha_{2,L}=2$, and $\alpha_{2,N}=4$, then $\bar{a}\rightarrow 0$. The association probabilities can be given in closed-form as follows:
$\mathcal{A}_{2,L}{=}1 - e^{-\pi p_L  d^2 \lambda_2}$,
$\quad\quad\mathcal{A}_{2,N}{=}
	e^{d^2 \pi (-p_L + p_N) \lambda_2} (1 - 
   \pi p_N d^2 \lambda_2 \:_1F_1[1; 
     2; \pi p_N d^2 \lambda_2]),$
     
\noindent where $\:_1F_1[a;b;z]$ is the Kummer Confluent Hypergeometric function.
\end{Corollary}
An interesting insight from {\bf Corollary~2} can be seen when the intensity of SBSs  $\lambda_2 \rightarrow \infty$ or $d$ is large, the probability of association to LOS SBSs $\mathcal{A}_{2,L}$ becomes almost 1. On the other hand, when $\lambda_2 \rightarrow 0$ or $d$ is small, $\:_1F_1[a;b;0]=1$ thus $\mathcal{A}_{2,N}$ becomes almost 1.

\subsection{Formulation of the Meta distribution, CSP and its $b^{\mathrm{th}}$ Moment in the Hybrid Network}\label{meta_dual}
When a {\color{black}device} associates with a mm-wave SBS, the overall CSP depends on the CSPs of the SIR and SNR on both the backhaul link and the access link, respectively. On the other hand, when a {\color{black}device} associates to MBS the CSP  depends on the SIR of the direct link. It is thus necessary to formulate the relationship between the meta distribution, CSP, and its $b^{\mathrm{th}}$ moment in the considered hybrid network as follows.
\begin{lemma}[Meta Distribution of the Typical {\color{black}device} in the Hybrid Network]\label{lem_3}
The combined meta distribution of the {\color{black} SIR/SNR} in the hybrid spectrum network can be characterized as follows:
\begin{equation}\label{Gil-pelaez-2hops}
\bar{F}_{P_{s,\text{T}}}(x)=\frac{1}{2}+\frac{1}{\pi}\int_{0}^{\infty}\frac{\Im\left(e^{-jt\log{x}}M_{jt, \text{T}}(\cdot)\right)}{t}\text{d}t,
\end{equation}
where $	M_{jt,\text{T}}(\theta)$ can be characterized by deriving  the  $b^{\mathrm{th}}$ moment of the $P_{s,\text{T}}(\cdot)$\footnote{The $b^{\mathrm{th}}$ moment of a random variable $X$ is the expected value of random variable to the power $b$, i.e., $\mathbb{E}[X^{b}]$.}.
\small
\begingroup \allowdisplaybreaks
\begin{align}\label{total_meta}
	M_{b,\text{T}}(\cdot) &= 
	{M_{b,\mathrm{Dual-Hop}}}+{M_{b,\mathrm{Single-Hop}} }\stackrel{(a)}{=}   \mathbb{E}_{\Phi}[\bar{\mathcal{A}}_{2}(l_{2}) P^b_{s,\mathrm{Dual-Hop}}(\theta_{2})] + \mathbb{E}_{\Phi}[\bar{\mathcal{A}}_1(l_1) P^b_{s,1}(\theta_{\mathcal{D}})],
	%%%%%
	\notag\\&\stackrel{(b)}{=} \mathbb{E}_{\Phi}\bigg[\bar{\mathcal{A}}_{2}(l_{2}) (P_{s,\mathrm{BH}}(\theta_{2})
	P_{s,2}(\theta_{\mathcal{D}}))^b \bigg]+\mathbb{E}_{{\Phi}}\bigg[ \bar{\mathcal{A}}_1(l_1)
	 P^b_{s,1}(\theta_{\mathcal{D}})\bigg], 
	%%%%
	 \notag\\&\stackrel{(c)}{=}\mathbb{E}_{{\Phi}}\bigg[P_{s,\mathrm{BH}}(\theta_{2})^{b}\bigg] \mathbb{E}_{{\Phi}}\bigg[ \bar{\mathcal{A}}_{2}(l_{2}) P_{s,2}(\theta_{\mathcal{D}})^{b}\bigg]+\mathbb{E}_{\Phi}\bigg[ \bar{\mathcal{A}}_1(l_1) P^b_{s,1}(\theta_{\mathcal{D}})\bigg],
	 %%%%
	 \notag\\&\stackrel{(d)}{=}\mathbb{E}_{{\Phi}}\bigg[P_{s,\mathrm{BH}}(\theta_{2})^{b}\bigg] \mathbb{E}_{{\Phi}}\bigg[ (\bar{\mathcal{A}}_{2,L}(l_{2,L})+\bar{\mathcal{A}}_{2,N}(l_{2,N})) P_{s,2}(\theta_{\mathcal{D}})^{b}\bigg]+\mathbb{E}_{\Phi}\bigg[ \bar{\mathcal{A}}_1(l_1) P^b_{s,1}(\theta_{\mathcal{D}})\bigg],
	 %%%%
	 \notag\\&\stackrel{(e)}{=} 
	\underbrace{ M_{b,\mathrm{BH}}(\theta_{2}) M_{b,2}(\theta_{\mathcal{D}})}_\text{Device Associated with SBS}+\underbrace{M_{b,1}(\theta_{\mathcal{D}})}_\text{Device Associated with MBS},
\end{align}
\endgroup
\normalsize
where ${M_{b,\mathrm{Dual-Hop}}}$ is the  $b^{\mathrm{th}}$ moment of the SIR/SNR when a {\color{black}device} associates to mm-wave SBS for dual-hop transmission and $M_{b,\mathrm{Single-Hop}}$ is the $b^{\mathrm{th}}$ moment of the SIR when a {\color{black}device} associates to MBS for direct transmission. After reformulation, we define $M_{b,\mathrm{BH}}(\theta_{2})$
as the unconditional $b^{\mathrm{th}}$ moment of the backhaul SIR, $M_{b,2}(\theta_{\mathcal{D}})$ as the unconditional  $b^{\mathrm{th}}$ moment of the SNR at access link when a {\color{black}device} associates to mm-wave SBS, and  $M_{b,1}(\theta_{\mathcal{D}})$ as the unconditional $b^{\mathrm{th}}$ moment of the SIR at direct link when a {\color{black}device} associates to $\mu$wave BS.
%\begin{align}\label{meta-55}
%\bar{F}_{P_{s,\text{T}}}(x)&\overset{\Delta}{=}\mathcal{A}_{2}\bar{F}_{P_{s,\rm{Dual-Hop}}}(x)+\mathcal{A}_{1}\bar{F}_{P_{s,\rm{Single-Hop}}}(x),
%\nonumber\\&=
%\mathcal{A}_{2}\mathbb{P}^{0}( P_{s,\mathrm{BH}}(\theta_{2})P_{s,2}(\theta_{\mathcal{D}})>x)
%+\mathcal{A}_{1}\mathbb{P}^{0}( P_{s,1}(\theta_{\mathcal{D}})>x) , \quad\quad x\in
%[0,1].
%\end{align}
Note that
$P_{s,1}(\theta_{\mathcal{D}})\overset{\Delta}{=}\mathbb{P}(\mathrm{SIR_{1,\mathcal{D}}}>\theta_{\mathcal{D}}|\mathbf{\Phi}_1,\text{tx})$ denotes the CSP of {\color{black}device} over the direct link,
$P_{s,\mathrm{BH}}(\theta_{2})\overset{\Delta}{=}\mathbb{P}(\mathrm{SIR_{1,2}}>\theta_{2}|\mathbf{\Phi_1},\text{tx})
$ denotes the CSP at backhaul link, and
$
P_{s,2}(\theta_{\mathcal{D}})\overset{\Delta}{=}\mathbb{P}(\mathrm{SNR_{2,\mathcal{D}}}>\theta_{\mathcal{D}}|\mathbf{\Phi}_2,\text{tx})
$ denotes the CSP for the access link transmission.   
\end{lemma}
\begin{proof}
Step (a) follows from the fact that the $b^{\mathrm{th}}$ moment of the SIR or SNR of a {\color{black}device} associated to tier $i$ can be defined as $M_b^{(i)}=\mathbb{E}[\bar{\mathcal{A}_i} M_{b|i}]$ where $\bar{\mathcal{A}_i}$ is the conditional association probability to tier $i$ and $M_{b|i}= P^b_{s,i}$ is the conditional $b^{\mathrm{th}}$ moment of the SIR or SNR in tier $i$. In our case, we have $\bar{\mathcal{A}}_{2}(l_{2})$ which is the conditional association probability to mm-wave SBS where $l_2 \in \{L,N\}$ since a {\color{black}device} can associate to either LOS or NLOS mm-wave SBS.
The step (b) follows from the fact that the CSP of the dual-hop transmission depends on the CSP of access and backhaul link; therefore, we have a product of the access and backhaul CSPs, i.e., $P_{s,\mathrm{BH}}(\theta_{2})P_{s,2}(\theta_{\mathcal{D}})$ that are independent random variables. There is no correlation since $\mu$wave backhaul does not interfere with mm-wave transmissions. The step (c) follows from the fact  if $X$ and $Y$ are
independent then $\mathbb{E}[(XY)^{b}]=\mathbb{E}[X^{b}]\mathbb{E}[Y^{b}]$. Finally, the step (d) follows from the definition of $\bar{\mathcal{A}}_{2}(l_{2})$ in {\bf Lemma~2} and the step (e) follows by applying the definition of moments.
\end{proof}

In the next section, we derive the CSP of access, backhaul, and direct links along with their respective $b^{\mathrm{th}}$ moments, as needed in {\bf Lemma~4} to characterize the overall moment as well as the meta distribution. 
	
%%%%%%%%%%%%%%%%%%%%%%%%%%%%%%%%%%%%%%%%%%%%%%%%%%%%%%%%%%%%%%%%%%%%%%%%%%%%%%%%%%%%%%%%%%%%%%%%%%%%%%%%%%%%%%%%%%%%%%
\section{Characterization of the CSPs and Moments}\label{v_sec}
 In this section, we  derive the CSPs $P_{s,\mathrm{BH}}(\theta_{2}), P_{s,2}(\theta_{\mathcal{D}}), P_{s,1}(\theta_{\mathcal{D}})$ and the  $b^{\mathrm{th}}$ moments $M_{b,\mathrm{BH}}(\theta_{2})$, $M_{b,2}(\theta_{\mathcal{D}})$, and $M_{b,1}(\theta_{\mathcal{D}})$ for  backhaul link, access link, and direct link, respectively.

\subsection{CSP and the $b^{\mathrm{th}}$ Moment - Access Link}
%%%%%%%%%%%%%%%%%%%%%%%%%%%%%%%%%%%%%%
We condition on having a {\color{black}device} at the origin which becomes a typical {\color{black}device}. 
The CSP of a typical {\color{black}device} at the origin associating with the  mm-wave SBS-tier (when $k=2$) can be described as follows:
\begin{equation}\label{total-coverage-UE}
P_{s,2}(\theta_{\mathcal{D}})=p_{L} P_{s,2,L}(\theta_{\mathcal{D}})+p_{N}P_{s,2,N}(\theta_{\mathcal{D}}).
\end{equation}

The {CSP of the SNR of a {\color{black}device} on the access link} with LOS can be defined  by substituting $\mathrm{SNR}_{2,\mathcal{D}}$  defined in Eq. \eqref{snr-access-O-mmwave} into Definition~1 as follows:
\begingroup
\allowdisplaybreaks
\begin{align}\label{MBS-OR-SBS-LOS}
P_{s,2,L}(\theta_{\mathcal{D}})&=\mathbb{P}\left(h_{L}(0, \mathbf{y}_{2,0})>\frac{\theta_{\mathcal{D}}r_{2,\mathcal{D}}^{\alpha_{2,L}}\sigma_{2}^{2}}{P_{2}G_{2}}|\mathbf{\Phi}_{1},\mathbf{\Phi}_{2},\text{tx}\right)
\stackrel{(a)}{=}1-\frac{\gamma\left(m_{L},\frac{m_{L}}{\Omega_{L}}\nu_{L}\right)}{\Gamma{(m_{L})}}
%&=\frac{\Gamma{(m_{L})}-\gamma\left(m_{L},\frac{m_{L}}{\Omega_{L}}\nu_{L}\right)}{\Gamma{(m_{L})}}
\stackrel{(b)}{=}\frac{\Gamma\left(m_{L},\frac{m_{L}}{\Omega_{L}}\nu_{L}\right)}{\Gamma{(m_{L})}},
%&\stackrel{(c)}{=}\widehat{\Gamma}\left(m_{L},\frac{m_{L}}{\Omega_{L}}\frac{\theta_{2}d_{\mathcal{D},2}^{\alpha_{2,L}}\sigma_{2}^{2}}{P_{2}G_{2}}\right).
\end{align}
\endgroup
where (a) follows from the definition of $\nu_{L}\overset{\Delta}{=}\frac{\theta_{\mathcal{D}}r_{2,\mathcal{D}}^{\alpha_{2,L}}\sigma_{2}^{2}}{P_{2}G_{2}}$ and the fact that the channel gain $h_{L}(0, \mathbf{y}_{2,0})$ is a normalized gamma random variable  and $\gamma(.,.)$ is the lower incomplete gamma function and  $\Gamma(s)=\gamma(s,x)+\Gamma(s,x)$, where $\Gamma(.,.)$ is the upper incomplete gamma function. % and (c) follows from the definition of the normalized incomplete gamma function $\widehat{\Gamma}(s,x)=\Gamma(s,x)/\Gamma(s)$.
Similarly,  {CSP of the SNR on the access link} for NLOS case can be given as follows:
\begingroup
\allowdisplaybreaks
\begin{align}\label{MBS-OR-SBS-NLOS}
P_{s,2,N}(\theta_{\mathcal{D}})
%&=1-\frac{\gamma\left(m_{N},\frac{m_{N}}{\Omega_{N}}\frac{\theta_{2}d_{\mathcal{D},2}^{\alpha_{2,N}}\sigma_{2}^{2}}{P_{2}G_{2}}\right)}{\Gamma{(m_{N})}},\\
%&=\frac{\Gamma{(m_{N})}-\gamma\left(m_{N},\frac{m_{N}}{\Omega_{N}}\frac{\theta_{2}d_{\mathcal{D},2}^{\alpha_{2,N}}\sigma_{2}^{2}}{P_{2}G_{2}}\right)}{\Gamma{(m_{N})}},\\
&=\frac{\Gamma\left(m_{N},\frac{m_{N}}{\Omega_{N}}\nu_{N}\right)}{\Gamma(m_{N})},
\end{align}
\endgroup
where $\nu_{N}\overset{\Delta}{=}\frac{\theta_{\mathcal{D}}r_{2,\mathcal{D}}^{\alpha_{2,N}}\sigma_{2}^{2}}{P_{2}G_{2}}$. 
As such, the {$\boldsymbol{b^{th}}$ moment of the CSP on the access link} for the typical {\color{black}device} when it is served by the mm-wave SBS tier is given by the following:
\begingroup
\allowdisplaybreaks
\begin{lemma} \label{3}
The $b^{\mathrm{th}}$ moment of the SNR at an ``access link'' when a {\color{black}device} associates with a mm-wave SBS 
can be characterized as follows: 
%in Eq. (\ref{m_b_2}) at the top of next page,
%\begin{figure*}[ht]
\small
	\begin{align}\label{m_b_2} &M_{b,\text{2}}(\theta_{\mathcal{D}})=\notag\\&\sum_{k=0}^{b}\binom{b}{k}(-1)^{k}\left(p_{L}^{b}\sum_{\ddot{k}=0}^{m_{L}k}\binom{m_{L}k}{\ddot{k}}(-1)^{\ddot{k}}\int_{0}^{d^{\alpha_{2,L}}}e^{-\zeta_{L}\ddot{k}\ddot{\nu_{L}} l_{\text{2,L}}} \bar{\mathcal{A}}_{2,L}(l_{2,L})+p_{N}^{b}\sum_{\ddot{k}=0}^{m_{N}k}\binom{m_{N}k}{\ddot{k}}(-1)^{\ddot{k}}\int_{d^{\alpha_{2,N}}}^{\infty}e^{-\zeta_{N}\ddot{k}\ddot{\nu_{N}} l_{\text{2,N}}} \bar{\mathcal{A}}_{2,N}(l_{2,N})\right),
	\end{align}
	\normalsize
	%\hrule
%	\end{figure*}
	where $\bar{\mathcal{A}}_{2,L}(l_{2,L})$ and $\bar{\mathcal{A}}_{2,N}(l_{2,N})$  are given in {\bf Lemma~2},
$\zeta_{L}\overset{\Delta}{=}m_{L}(m_{L}!)^{-1/m_{L}}$, $\nu_{L}\overset{\Delta}{=}\frac{\theta_{\mathcal{D}}r_{2,\mathcal{D}}^{\alpha_{2,L}}\sigma_{2}^{2}}{P_{2}G_{2}}$, $\zeta_{N}\overset{\Delta}{=}m_{N}(m_{N}!)^{-1/m_{N}}$, and $\nu_{N}\overset{\Delta}{=}\frac{\theta_{\mathcal{D}}r_{2,\mathcal{D}}^{\alpha_{2,N}}\sigma_{2}^{2}}{P_{2}G_{2}}$, $\ddot{\nu_{L}}\overset{\Delta}{=}\frac{\nu_{L}}{r^{\alpha_{2,L}}}=\frac{\nu_{L}}{l_{\text{2,L}}}=\frac{\theta_{\mathcal{D}}\sigma_{2}^{2}}{P_{2}G_{2}}$  and $\ddot{\nu_{N}}\overset{\Delta}{=}\frac{\nu_{N}}{r^{\alpha_{2,N}}}=\frac{\nu_{N}}{l_{\text{2,N}}}=\frac{\theta_{\mathcal{D}}\sigma_{2}^{2}}{P_{2}G_{2}}$.
\end{lemma}
\endgroup
\begin{proof}
	See {\bf Appendix~C}.
\end{proof}
For  $\alpha_{1}=4$, $\alpha_{2,L}=2$, and $\alpha_{2,N}=4$,  we can get $M_{b,\text{2} }(\theta_{\mathcal{D}})$ in closed-form using {\bf Corollary~1}. Also, for scenarios where $\mathcal{N}\rightarrow \infty$, $\alpha_{1}=4$, $\alpha_{2,L}=2$, and $\alpha_{2,N}=4$, then $\bar{a}\rightarrow 0$. Also, $\ddot{v}_L \rightarrow 0$ and $\ddot{v}_N \rightarrow 0$,  we can get $M_{b,\text{2} }(\theta_{\mathcal{D}})$ in closed-form using {\bf Corollary~2}.

\subsection{CSP and Moment - Backhaul Link}
%%%%%%%%%%%%%%%%%%%%%%%%%%%%%%%%%%%%%%
%%%%%%%%%%%%%%%%%%%%%%%%%%%%%%%%%%%%%%%%%%%%%%%%%%%%%%%%%%%%%
For the backhaul link,  we condition on having a SBS at the origin which becomes the typical SBS. Using the expression of  $\mathrm{SIR}_{1,2}$ in  Eq. (\ref{SIR-backhaul-O})  {the CSP of the backhaul link} $P_{s,\mathrm{BH}}(\theta_{2})$ can be given as:
\begingroup \allowdisplaybreaks
\begin{align}\label{P_s_b}
P_{s,\mathrm{BH}}(\theta_{2})&=\mathbb{P}\left(g(0, \mathbf{y}_{1,0})>\frac{\theta_{2}r_{1,2}^{\alpha_1}}{P_{1}}\mathcal{I}_{1,2}|\mathbf{\Phi}_{1},\mathbf{\Phi}_{2},\text{tx}\right)
%%%
\nonumber\stackrel{(a)}{=}\mathbb{E}\bigg[\exp(-\theta_{2}r_{1,2}^{\alpha_1}\sum_{i:\mathbf{y}_{1,i}\in\mathbf{\Phi}_1\backslash \{\mathbf{y}_{1,0}\}}\|\mathbf{y}_{1,i}\|^{-\alpha_1}g(0,\mathbf{y}_{1,i}))\bigg],
%%%
\notag\\
&=\prod_{\mathbf{y}_{1,i}\in\mathbf{\Phi}_1\backslash \{\mathbf{y}_{1,0}\}}\mathbb{E}\bigg[\exp\left(-\theta_{2}r_{1,2}^{\alpha_1}\|\mathbf{y}_{1,i}\|^{-\alpha_1}g(0,\mathbf{y}_{1,i})\right)\bigg]
%%%
\stackrel{(b)}{=}\prod_{\mathbf{y}_{1,i}\in\mathbf{\Phi}_1\backslash \{\mathbf{y}_{1,0}\}}\frac{1}{1+\theta_{2}\left(\frac{r_{1,2}}{\|\mathbf{y}_{1,i}\|}\right)^{\alpha_1}}.
\end{align}
 \endgroup
\normalsize
where (a) follows from the Rayleigh fading channel gain $g(0, \mathbf{y}_{1,0})\sim \exp(1)$  and (b) is found by taking the expectation with respect to $g(0,\mathbf{y}_{2,i})$. The {$\boldsymbol{b^{th}}$ moment of the CSP on the backhaul link} is given as:
\begin{align}\label{two-hops-b-backhaul}
M_{b,\mathrm{BH}}(\theta_{2})&=\mathbb{E}\bigg[P_{s,\mathrm{BH}}(\theta_{2})^{b}\bigg]=\mathbb{E}\bigg[\prod_{\mathbf{y}_{1,i}\in\mathbf{\Phi}_1\backslash \{\mathbf{y}_{1,0}\}}\frac{1}{\left(1+\theta_{2}\left(\frac{r_{1,2}}{\|\mathbf{y}_{1,i}\|}\right)^{\alpha_1}\right)^{b}}\bigg],\notag\\
&\stackrel{(a)}{=}\left({1+2\int_{0}^{1}\left(1-\frac{1}{(1+\theta_{2}r^{\alpha_1})^{b}}\right)r^{-3}dr}\right)^{-1}{=}
\frac{1}{\text{ }_{2}F_{1}(b,-\frac{2}{\alpha_1};1-\frac{2}{\alpha_1};-\theta_{2})},
\end{align}
\normalsize
where (a) follows from the probability generating functional (PGFL) of PPP, i.e., $
G_{\mathcal{R}}[f]\overset{\Delta}{=}\mathbb{E}\prod_{x\in\mathcal{R}}f(x)=\frac{1}{1+2\int_{0}^{1}(1-f(x))x^{-3}dx}.$ \cite[lemma 1]{ganti2016asymptotics} and ${\text{}_{2}F_{1}(.,.;.;.)}$ represents Gauss` Hyper-geometric function.

%%%%%%%%%%%%%%%%%%%%%%%%%%%%%%%%%%%%%%%%%%%%%%%%%%%%%%%%%%%%%%%%%%%%%%%%%%%%%%%%%%%%%%%%%%%%%%%%%%%%%%%%%%%%
\subsection{CSP and Moment - Direct Link}
 Using the expression of $\mathrm{SIR}_{1,\mathcal{D}}$  in Eq. (\ref{SIR-direct-link}), we calculate the {CSP of the direct link} $P_{s,1}(\theta_{\mathcal{D}})$ as follows:
\begingroup
\allowdisplaybreaks
\small
\begin{align}\label{MBS-OR-SBS}
P_{s,1}(\theta_{\mathcal{D}})&=\mathbb{P}\left(g(0, \mathbf{y}_{1,0})>\frac{\theta_{\mathcal{D}}r_{1,\mathcal{D}}^{\alpha_{1}}}{P_{1}}\mathcal{I}_{1,\mathcal{D}}|\mathbf{\Phi}_{1},\mathbf{\Phi}_{2},\text{tx}\right)\stackrel{(a)}{=}\mathbb{E}\bigg[\exp\left(-\theta_{\mathcal{D}}r_{1,\mathcal{D}}^{\alpha_{1}}\sum_{i:\mathbf{y}_{1,i}\in\mathbf{\Phi}_1\backslash \{\mathbf{y}_{1,0}\}}\|\mathbf{y}_{1,i}\|^{-\alpha_{1}}g(0,\mathbf{y}_{1,i})\right)\bigg],\notag\\
&=\prod_{\mathbf{y}_{1,i}\in\mathbf{\Phi}_1\backslash \{\mathbf{y}_{1,0}\}}\mathbb{E}\bigg[\exp\left(-\theta_{\mathcal{D}}r_{1,\mathcal{D}}^{\alpha_{1}}\|\mathbf{y}_{1,i}\|^{-\alpha_{1}}g(0,\mathbf{y}_{1,i})\right)\bigg]
\stackrel{(b)}{=}
\prod_{\mathbf{y}_{1,i}\in\mathbf{\Phi}_1\backslash \{\mathbf{y}_{1,0}\}}\frac{1}{1+\theta_{\mathcal{D}}\left(\frac{r_{1,\mathcal{D}}}{\|\mathbf{y}_{1,i}\|}\right)^{\alpha_1}},
\end{align}
\endgroup
\normalsize
where (a) follows from the channel gain $g(0, \mathbf{y}_{1,0})\sim \exp(1)$ and is independently exponentially distributed with unit mean and (b) is obtained by taking the expectation with respect to $g(0,\mathbf{y}_{1,i})$. While taking the association probabilities into consideration, the {$\boldsymbol{b^{th}}$ moment of the CSP} $P_{s,1}(\theta_{\mathcal{D}})$ of the typical {\color{black}device} when it is served by a $\mu$wave MBS is characterized in the following lemma.
\begin{lemma}[The $b^{\mathrm{th}}$ moment of the CSP ($P_{s,1}(\theta_{\mathcal{D}})$) when a {\color{black}device} associates with a MBS]\label{4}
	The $b^{\mathrm{th}}$ moment of the CSP experienced by a {\color{black}device}, when the {\color{black}device} associates with a MBS, can be characterized as follows: 
	%can be characterized in Eq. (\ref{m_b_1}) at the top of the next page.
	%\begin{figure*}[ht!]
	\small
	\begingroup \allowdisplaybreaks
	\begin{align}\label{m_b_1}
	&M_{b,\text{1}}(\theta_{\mathcal{D}})=
	\frac{2\pi\lambda_{1}}{\hat{a}\alpha_{1}}\Bigg\{\int_{0}^{d^{\alpha_{2,L}}}H(l_1)\exp\left(-\pi\lambda_{2}p_{L}l_{1}^{\frac{2}{\alpha_{2,L}}}\right)d\text{l}_{1}+\int_{d^{\alpha_{2,L}}}^{d^{\alpha_{2,N}}}H(l_1)\exp\left(-\pi\lambda_{2}p_{L}d^{2}\right)d\text{l}_{1}+\notag\\
	&\quad\quad\quad\quad\int_{d^{\alpha_{2,N}}}^{\infty}H(l_1)\exp\left(-\pi\lambda_{2}[p_{L}d^{2}+p_{N}\left(l_{1}^{\frac{2}{\alpha_{2,N}}}-d^{2}\right)\right)d\text{l}_{1}\Bigg\}\times
	\exp
	\left(\frac{-2\lambda_{1}\pi l_{1}^{\frac{2}{\alpha_{1}}}}{\alpha_{1}}\int_{0}^{1}\bigg[1-\frac{1}{\left(1+\theta_{\mathcal{D}}v\right)^{b}}\bigg]\frac{1}{v^{\frac{2}{\alpha_{1}}+1}}\text{d}v\right),
	\end{align}
%	\hrule
	\endgroup
	\normalsize
%	\end{figure*}
\end{lemma}
\begin{proof}
	See {\bf Appendix~D}.
\end{proof}
Note that $\int_{0}^{1}\bigg[1-\frac{1}{\left(1+\theta_{\mathcal{D}}v\right)^{b}}\bigg]\frac{1}{v^{\frac{2}{\alpha_{1}}+1}}\text{d}v$ is independent of $l_1$, thus  where $\mathcal{N}\rightarrow \infty$ or $\alpha_{1}=4$, $\alpha_{2,L}=2$, and $\alpha_{2,N}=4$, then  we can get a closed-form for the three integral over $l_1$ using {\bf Corollary~1} and {\bf Corollary~2}.

\subsection{Combined $b^{\mathrm{th}}$ Moment of the CSP in Hybrid  Networks}\label{exact_meta}
After substituting the values of $M_{b,\mathrm{BH}}(\theta_{2})$, $M_{b,2}(\theta_{\mathcal{D}})$, and $M_{b,1}(\theta_{\mathcal{D}})$ in Eq. \eqref{two-hops-b-backhaul}, Eq. \eqref{m_b_2}, and Eq. (\ref{m_b_1}), respectively into the total meta distribution for the entire  network in Eq. (\ref{total_meta}), we get the $b^{\mathrm{th}}$ moment of the CSP at a typical {\color{black}device} as follows:
%as shown in Eq. (\ref{whole-moments}) at the top of the next page.
%\begin{figure*}[ht!]
\small
  \begin{align}\label{whole-moments}
        &M_{b,\text{T}}
        =\frac{1}{\text{ }_{2}F_{1}(b,-\frac{2}{\alpha_1};1-\frac{2}{\alpha_1};-\theta_{2})}\times
                 \bigg\{\sum_{k=0}^{b}\binom{b}{k}(-1)^{k}\left(p_{L}^{b}\sum_{\ddot{k}=0}^{m_{L}k}\binom{m_{L}k}{\ddot{k}}(-1)^{\ddot{k}}\int_{0}^{d^{\alpha_{2,L}}}e^{-\zeta_{L}\ddot{k}\ddot{\nu_{L}} l_{\text{2,L}}} \bar{\mathcal{A}}_{2,L}(l_{2,L})+\right.\notag\\&\quad\quad\quad\quad\quad\quad\quad\quad\quad\quad\quad\quad\quad\quad\quad\quad\left.p_{N}^{b}\sum_{\ddot{k}=0}^{m_{N}k}\binom{m_{N}k}{\ddot{k}}(-1)^{\ddot{k}}\int_{d^{\alpha_{2,N}}}^{\infty}e^{-\zeta_{N}\ddot{k}\ddot{\nu_{N}} l_{\text{2,N}}} \bar{\mathcal{A}}_{2,N}(l_{2,N})\right)\bigg\}+M_{b,1}(\theta_{\mathcal{D}}),
      \end{align}
      \normalsize
%\end{figure*}
%%%%%%%%%%%%%%%%%%%%%%%%%%%%%%%%%%%%%%%%%%%%%%%%%%%%%%%%%%%%%%%%%%%%%%%%%%%%%%%%%%%%%%%%%%%%%%%%%%%%%%%%%%%%
 
In the next section, we  use the combined $b^{\mathrm{th}}$ moment in \eqref{whole-moments} to compute the meta distributions of SIR/SNR and data rate using Gil-Pelaez inversion and the Beta approximation.

 %%%%%%%%%%%%%%%%%%%%%%%%%%%%%%%%%%%%%%%%%%%%%%%%%%%%%%%%%%%%%%%%%%%%%%%%%%%%%%%%%%%%%%%%%%%%%  

 \section{Computing the Meta Distributions and Special Cases}\label{vi_sec}   
 In this section, we compute the meta distribution of SIR/SNR  using Gil-Pelaez inversion and beta approximation by applying the derived result of $M_{b,\text{T}}$. Special cases where $b=1$ provides the standard coverage probability and $b=-1$ provides the mean local delay are discussed.  Further, we show how to evaluate the data rate meta distribution from the derived framework. 
 
 %%%%%%%%%%%%%%%%%%%%%%%%%%%%%%%
 \subsection{Computing the Meta Distribution of SIR/SNR}
 Technically, substituting $b=jt$ in \eqref{whole-moments}, we should obtain the imaginary moments $M_{jt, \text{T}}$. However, since the expression of $M_{jt, \text{T}}$ relies on a Binomial expansion of power $b$, the results cannot be obtained directly through substitution. Therefore, we apply Newton's generalized binomial theorem given as follows:
 \begin{Definition}\label{def_1}
Isaac Newton`s generalized binomial theorem is to allow real exponents other than non-negative integers, i.e., imaginary exponent $r$, as
$
\binom{r}{k}=\frac{r(r-1)...(r-k+1)}{k!}=\frac{(r)_{k}}{k!},
$
where $(.)_{k}$ is the Pochhammer symbol, which stands here for a falling factorial.
\end{Definition}
Applying {\bf Definition~2} in step~(e) of {\bf Appendix~C}, we then  obtain the expression for $M_{jt, \text{T}}$ as follows:
%as shown in Eq. (\ref{mmwave_imaginary_moment}) at the top of the next page.
%\begin{figure*}[ht!]
\small
      \begin{align}\label{mmwave_imaginary_moment}
      &M_{jt,\text{T}}=\frac{1}{\text{ }_{2}F_{1}(jt,-\frac{2}{\alpha_1};1-\frac{2}{\alpha_1};-\theta_{2})}\times
                 \bigg\{p_{L}^{jt}
                 \sum_{k=0}^{\infty}\frac{(jt)_{k}}{k!}(-1)^{k}\sum_{\ddot{k}=0}^{m_{L}k}\binom{m_{L}k}{\ddot{k}}(-1)^{\ddot{k}}\int_{0}^{d^{\alpha_{2,L}}} e^{-\zeta_{L}\ddot{k}\ddot{\nu_{L}} l_{\text{2,L}}} \bar{\mathcal{A}}_{2,L}(l_{2,L})+\notag\\
                 &\quad\quad\quad\quad\quad\quad\quad\quad\quad\quad p_{N}^{jt}\sum_{k=0}^{\infty}\frac{(jt)_{k}}{k!}(-1)^{k}\sum_{\ddot{k}=0}^{m_{N}k}\binom{m_{N}k}{\ddot{k}}(-1)^{\ddot{k}}\int_{d^{\alpha_{2,N}}}^{\infty} e^{-\zeta_{N}\ddot{k}\ddot{\nu_{N}} l_{\text{2,N}}} \bar{\mathcal{A}}_{2,N}(l_{2,N})\bigg\}+M_{jt,1}(\theta_{\mathcal{D}}),
          \end{align}
          \normalsize
 %         \hrule
 %         \end{figure*}
          
          The imaginary moments can be substituted in the Gil-Pelaez inversion theorem as in {\bf Definition~1} to obtain $\bar{F}_{P_s,\text{T}}$.
%The meta distribution expression derived using the Gil-Pelaez in Eq. \eqref{Gil-pelaez-2hops} is exact and it is difficult to evaluate numerically since it usually involves many calculations of imaginary moments, which prohibits direct insights into the meta distribution and its applications in mapping to other performance metrics like the ergodic data rate \cite{deng2017fine}. 
Furthermore, we follow \cite{haenggi2016meta,salehi2018meta,wang2017meta} to approximate the meta distribution by a Beta distribution by matching the first and second moments, which are easily obtained from the general result in Eq. (\ref{whole-moments}) by substituting $b=1$ and $b=2$ to get $M_{1,\text{T}}$ and $M_{2,\text{T}}$, respectively. %as shown in Eq. (\ref{first-momemnt-beta}) and (\ref{second-momemnt-beta}) in the top of this page.
Taking $\beta\overset{\Delta}{=}\frac{(M_{1,\text{T}}-M_{2,\text{T}})(1-M_{1,\text{T}})}{M_{2,\text{T}}-M_{1,\text{T}}^{2}}$, the meta distribution using beta approximation can be given as follows:
\begin{equation}\label{beta-match1}
\bar{F}_{P_{s,\text{T}}}(x)\approx 1-I_{x}\left(\frac{\beta M_{1,\text{T}}}{1-M_{1,\text{T}}},\beta\right), \quad\quad x\in [0,1],
\end{equation}

\subsection{Mean and Variance of the Local Delay}
The mean local delay is the mean number of transmission attempts, i.e., re-transmissions,  needed to successfully transmit a packet to the target receiver. The mean local delay $M_{-1,\text{T}}$  which is the $-1^{\mathrm{st}}$ moment of the  CSP of a typical {\color{black}device} should be calculated by substituting $b = -1$ in Eq. (\ref{whole-moments}).  However, since the expression of $M_{b, \text{T}}$ relies on a Binomial expansion of power $b$, the results cannot be obtained directly through substitution. Therefore, we apply  Binomial theorem for the negative integers as follows: 
\begin{Definition}\label{def_2} The Binomial theorem for a negative integer power $n$ can be given \cite{kronenburg2011binomial} as
$
(x+y)^{n}=\sum_{k=0}^{\infty}(-1)^{k}\binom{-n+k-1}{k}y^{n-k}x^{k},
$
\end{Definition}
Applying {\bf Definition~3} in step~(e) of {\bf Appendix~C}, we then  obtain the expression for $M_{-1, \text{T}}$ as follows:
%in Eq. (\ref{local_delay_total}) at the top of the next page.
%\begin{figure*}[ht!]
\small
	\begin{align}\label{local_delay_total}
	&M_{-1,\text{T}}=\frac{1}{\text{ }_{2}F_{1}(-1,-\frac{2}{\alpha_1};1-\frac{2}{\alpha_1};-\theta_{2})}\times\bigg\{p_{L}^{-1} \sum_{k=0}^{\infty}\sum_{\ddot{k}=0}^{m_{L}k}\binom{m_{L}k}{\ddot{k}}(-1)^{\ddot{k}}\int_{0}^{d^{\alpha_{2,L}}} e^{-\zeta_{L}\ddot{k}\ddot{\nu_{L}} l_{\text{2,L}}} \bar{\mathcal{A}}_{2,L}(l_{2,L})+\notag\\ &\quad\quad\quad\quad\quad\quad\quad\quad\quad\quad\quad\quad\quad\quad\quad  p_{N}^{-1}\sum_{k=0}^{\infty}\sum_{\ddot{k}=0}^{m_{N}k}\binom{m_{N}k}{\ddot{k}}(-1)^{\ddot{k}}\int_{d^{\alpha_{2,N}}}^{\infty}
	e^{-\zeta_{N}\ddot{k}\ddot{\nu_{N}} l_{\text{2,N}}} \bar{\mathcal{A}}_{2,N}(l_{2,N}) dl_{2,N}\bigg\}+M_{-1,1}(\theta_{\mathcal{D}}),
	\end{align}
%	\hrule
	\normalsize
%	\end{figure*}
	
{\bf Remark:} In order to better characterize the fluctuation of the local
delay, the variance of the local delay also referred to as network jitter can be given by
$
NJ=M_{-2,\text{T}}-M_{-1,\text{T}}^{2}.
$
%%%%%%%%
\subsection{The Meta Distribution of the Data Rate in Hybrid Spectrum Networks}\label{data_rate_meat}
Let $\mathcal{T}$ denote the data rate (in bits/sec) of the typical {\color{black}device} on a specific transmission link which is a random variable and is defined as $\mathcal{R} = W \log_{2}(1 + \mathrm{SIR})$ using Shannon capacity.     Using the  meta distribution of the SIR, the meta distribution of the data rate can be derived to present the fraction of active devices in each realization of the point process that have a data rate $\mathcal{R}$ greater than $\mathcal{T}$ with probability at least $x$, i.e., {\color{black}devices} data rate reliability threshold. That is, first deriving the CSP of the data rate as follows:
\begin{align}
\mathbb{P}[\mathcal{R}>\mathcal{T}|\mathbf{\Phi},\text{tx}]&=\mathbb{P}[W\log_{2}(1+\mathrm{SIR})>\mathcal{T}|\mathbf{\Phi},\text{tx}]=\mathbb{P}[\mathrm{SIR}>2^{\frac{\mathcal{T}}{W}}-1|\mathbf{\Phi},\text{tx}],
\end{align}
where $P_{s}(2^{\frac{\mathcal{T}}{W}}-1)\overset{\Delta}{=}\mathbb{P}(\mathrm{SIR}>2^{\frac{\mathcal{T}}{W}}-1|\mathbf{\Phi}_1,\text{tx})$ denote the CSP of the {\color{black}device} data rate over single link. Finally, deriving the $b^{\mathrm{th}}$ moment of the CSP of the data rate and applying Gil-Pelaez inversion we can obtain the meta distribution of the data rate.

\begin{Corollary}
Similar to the meta distribution of the SIR/SNR derived in Lemma \ref{lem_3} and  conditioned on the location of the point process, we derive the meta distribution of the data rate in hybrid 5G cellular networks, using the moment $\mathcal{Q}_{b}$ of the conditional data rate as follows:
\small
\begin{align}\label{data_rate_bth_moment} &\mathcal{Q}_{b}(\mathcal{T})=\mathbb{E}[\bar{\mathcal{A}}_{2}(l_2)\mathbb{P}^{0}\left( P_{s,\mathrm{BH}}(2^{\frac{\mathcal{T}_{\mathrm{BH}}}{(1-\eta)W_{1}}}-1)P_{s,2}(2^{\frac{\mathcal{T}_2}{W_{2}}}-1)>x\right)]+\mathbb{E}[\bar{\mathcal{A}}_{1}(l_1)\mathbb{P}^{0}( P_{s,1}(2^{\frac{\mathcal{T}_{1}}{\eta W_1}}-1)>x)],\notag\\
&=M_{b,\mathrm{BH}}\left(2^{\frac{\mathcal{T}_{\mathrm{BH}}}{(1-\eta)W_{1}}}-1\right) M_{b,2}\left(2^{\frac{\mathcal{T}_2}{W_{2}}}-1\right)+M_{b,1}\left(2^{\frac{\mathcal{T}_{1}}{\eta W_1}}-1\right),
\end{align}
\normalsize
where $P_{s,1}(2^{\frac{\mathcal{T}_{1}}{\eta W_1}}-1)\overset{\Delta}{=}\mathbb{P}(\mathrm{SIR_{1,\mathcal{D}}}>2^{\frac{\mathcal{T}_{1}}{\eta W_1}}-1|\mathbf{\Phi}_1,\text{tx})$,
$P_{s,\mathrm{BH}}(2^{\frac{\mathcal{T}_{\mathrm{BH}}}{(1-\eta)W_{1}}}-1)\overset{\Delta}{=}\mathbb{P}(\mathrm{SIR_{1,2}}>2^{\frac{\mathcal{T}_{\mathrm{BH}}}{(1-\eta)W_{1}}}-1|\mathbf{\Phi_1},\text{tx})
$, and
$
P_{s,2}(2^{\frac{\mathcal{T}_2}{W_{2}}}-1)\overset{\Delta}{=}\mathbb{P}(\mathrm{SNR_{2,\mathcal{D}}}>2^{\frac{\mathcal{T}_2}{W_{2}}}-1|\mathbf{\Phi}_2,\text{tx})
$ denotes the CSP of the {\color{black}device} data rate at the direct, backhaul, and access link, respectively.
\end{Corollary}
In the following section, we  discuss the application of this framework in two scenarios (i) $\mu$wave only network and (ii) mm-wave backhauls and microwave access links.

\section{Extensions of The Model to Other Network Architectures}\label{mu-wave-meta}
%%%%%%%%%%%%%%%%%%%%%%%%%%%%%%%%%%%%%%%%%%%%%%%%%%%%%%%%%%%%%%%%%%%%%%%%%%%%%%%%%%%%%%%%%%
The framework discussed above can be flexibly applied to different network architectures.  In this section we discuss how to extend the framework to two other network architectures: 1) both tiers operating in the sub-6GHz (microwave) spectrum as in traditional cellular networks; and 2) the two tiers operating in two millimeter-wave spectrums which are orthogonal to each other. Due to space limitation, we  provide only general directions of how to extend the earlier framework to these two other network architectures.
\subsection{The Meta Distribution of the SIR in Microwave-only Cellular Networks}
%%%%%%%%%%%%%%%%%%%%%%%%%%%%%%%%%%%%%%%%%%%%%%%%%%%%%%%%%%%%%%%%%%%%%%%%%%%%%%%%%%%%%%%%%
We characterize the meta distribution of the downlink $\mathrm{SIR}$ attained at a typical {\color{black}device} in a  $\mu$wave-only cellular network, i.e., the access and backhaul links of SBSs operate in the $\mu$wave frequency. A {\color{black}device} associates with either a serving MBS for direct transmissions (when $k=1$) or a SBS for dual-hop transmissions (when $k=2$), depending on the biased received signal power criterion. MBSs and SBSs are assumed to operate on orthogonal spectrums; thus, there is no inter-tier interference. { On the other hand, each SBS associates with a MBS based on the maximum received power at the SBS. The association criterion for a typical {\color{black}device} can be described as follows \cite{jo2012heterogeneous}:}
      \begin{equation}\label{biasing1}
    P_{k}B_{k}(\min_{i}\|\mathbf{y}_{k,i}-\mathbf{x}\|)^{-\alpha_{k}}\ge P_{j}B_{j}(\min_{i^{'}}\|\mathbf{y}_{j,i^{'}}-\mathbf{x}\|)^{{-\alpha_{j}}}, \forall j
    \end{equation}
    \noindent
    where $\|.\|$ denotes the Euclidean distance. A typical {\color{black}device} associates with a serving node (given by Eq. (\ref{biasing1}))), which is termed the tagged SBS. For the sake of clarity, we define $\hat{P}_{jk}\overset{\Delta}{=}\frac{P_{j}}{P_{k}}$, $\hat{B}_{jk}\overset{\Delta}{=}\frac{B_{j}}{B_{k}}$, $\hat{\lambda}_{jk}\overset{\Delta}{=}\frac{\lambda_{j}}{\lambda_{k}}$. As derived in \cite{jo2012heterogeneous}, the conditional association probability for the typical {\color{black}device} connecting to the $k^{\mathrm{th}}$ tier (conditional over the desired link distance $r_{\mathcal{D},k}$) is as follows:
    \begin{equation}\label{condition-access}
    \mathbb{P}(n=k|r_{\mathcal{D},k})=\prod_{j\neq k}e^{-\pi\lambda_{j}(\hat{P}_{jk}\hat{B}_{jk})^{2/\alpha_{j}}r^{2}},
    \end{equation}
    where $n$ denotes the index of the tier associating with the typical {\color{black}device}. We calculate the {CSP $P_{s,2'}(\theta_{\mathcal{D}})$}  (when $k=2$) of the access link operating in the $\mu$wave band as follows:
     \begin{align}\label{MBS-OR-SBS_macrowave}
        &P_{s,2'}(\theta_{\mathcal{D}})=\mathbb{P}\left(g(0, \mathbf{y}_{2,0})>\frac{\theta_{\mathcal{D}}r_{2,\mathcal{D}}^{\alpha_{2}}}{P_{2}}\mathcal{I}_{2,\mathcal{D}}|\mathbf{\Phi}_{1},\mathbf{\Phi}_{2},\text{tx}\right),\notag\\
        &\stackrel{(a)}{=}\mathbb{E}\bigg[\exp\left(-\theta_{\mathcal{D}}r_{2,\mathcal{D}}^{\alpha_{2}}\sum_{i:\mathbf{y}_{2,i}\in\mathbf{\Phi}_2 \backslash \{\mathbf{y}_{2,0}\}}\|\mathbf{y}_{2,i}\|^{-\alpha_{2}}g(0,\mathbf{y}_{2,i})\right)\bigg]\stackrel{(b)}{=}\prod_{\mathbf{y}_{2,i}\in\mathbf{\Phi}_2\backslash \{\mathbf{y}_{2,0}\}}\frac{1}{1+\theta_{\mathcal{D}}\left(\frac{r_{2,\mathcal{D}}}{\|\mathbf{y}_{2,i}\|}\right)^{\alpha_2}}.
      \end{align}
      where (a) follows from the channel gain $g(0, \mathbf{y}_{2,0})\sim \exp(1)$ and is independently exponentially distributed with unit mean and (b) is obtained by taking the expectation with respect to $g(0,\mathbf{y}_{2,i})$.
      \begin{lemma}\label{5}
      Using Eq. (\ref{MBS-OR-SBS}) and Eq. (\ref{MBS-OR-SBS_macrowave}), we calculate a general expression for the {$\boldsymbol{b^{\mathrm{th}}}$ moment of the CSP on direct link} $M_{b,\text{k'}}$ (when $k=2$) and the {$\boldsymbol{b^{th}}$ moment of the CSP at access link} (when $k=1$) as:
    \small
    \begin{align}\label{MBS-or-SBS-b-moment}
    M_{b,\text{k'}}&=\frac{1}{\sum\limits_{j\neq k}\hat{\lambda}_{jk}(\hat{P}_{jk}\hat{B}_{jk})^{2/\alpha_{j}}+\text{ }_{2}F_{1}(b,-\frac{2}{\alpha_{k}};1-\frac{2}{\alpha_{k}};-\theta_{\mathcal{D}})}.
    \end{align}
    \normalsize
    \end{lemma}
\begin{proof}
See {\bf Appendix~E}.
\end{proof}
Note that Lemma \ref{5} is novel and different from \cite{rrrt} where we derive the $b^{\mathrm{th}}$ moment of CSP for orthogonal spectrum two tier network while the work in \cite{rrrt} is done for shared spectrum tiers.

Similarly, the moment of the CSP  of a typical {\color{black}device} with offloading biases is defined as follows:
\begin{align}\label{total_meta_microwave}
        M_{b,\text{T}}&= \underbrace{M_{b,\text{dual-hop}}}_\text{Dual-hop  transmission}+\underbrace{M_{b,1'}(\theta_{\mathcal{D}})}_\text{Direct transmission}\stackrel{(a)}{=} M_{b,\mathrm{BH}}(\theta_{2}) M_{b,2'}(\theta_{\mathcal{D}})+M_{b,1'}(\theta_{\mathcal{D}}),
\end{align}
\normalsize
where $M_{b,\mathrm{BH}}(\theta_{2})$, $M_{b,2'}(\theta_{\mathcal{D}})$, and $M_{b,1'}(\theta_{\mathcal{D}})$ are defined in Eq. \eqref{two-hops-b-backhaul}, Eq. \eqref{MBS-or-SBS-b-moment} (when $k=2$), and Eq. (\ref{MBS-or-SBS-b-moment}) (when $k=1$), respectively. The step (a) follows from the similar approach as taken in {\bf Lemma~4}.
\small
  \begingroup \allowdisplaybreaks
          \begin{align}\label{two-hops-b-moment}
            M_{b,\text{dual-hop}}&=\mathbb{E}\bigg[P_{s,\mathrm{BH}}(\theta_{2})^{b}\times \prod_{j\neq k}e^{-\pi\lambda_{j}(\hat{P}_{jk}\hat{B}_{jk})^{2/\alpha_{j}}r^{2}} P_{s,2'}(\theta_{\mathcal{D}})^{b}\bigg],\notag\\
            &\stackrel{(a)}{=}\underbrace{\mathbb{E}\bigg[P_{s,\mathrm{BH}}(\theta_{2})^{b}\bigg]}_\text{$M_{b,\mathrm{BH}}(\theta_{2})$(Backhaul link )}\underbrace{\mathbb{E}\bigg[\prod_{j\neq k}e^{-\pi\lambda_{j}(\hat{P}_{jk}\hat{B}_{jk})^{2/\alpha_{j}}r^{2}}P_{s,2'}(\theta_{\mathcal{D}})^{b}\bigg]}_\text{$M_{b,2}(\theta_{\mathcal{D}})$(access link)},\notag\\
                 &\stackrel{(b)}{=}\frac{1}{\text{ }_{2}F_{1}(b,-\frac{2}{\alpha_1};1-\frac{2}{\alpha_1};-\theta_{2})}\times\frac{1}{\hat{\lambda}_{12}(\hat{P}_{12}\hat{B}_{12})^{2/\alpha_{1}}+\text{ }_{2}F_{1}(b,-\frac{2}{\alpha_{2}};1-\frac{2}{\alpha_{2}};-\theta_{\mathcal{D}})},
          \end{align}
          \endgroup
          \normalsize
         where (a) follows from the independence between the location of the MBSs and SBSs. In step (b) we substitute  $M_{b,\text{BH}}(\theta_{2})$ from Eq. \eqref{two-hops-b-backhaul} and  $M_{b,\text{2'}}(\theta_{\mathcal{D}})$ into Eq. \eqref{MBS-or-SBS-b-moment} when $k=2$.
         By substituting Eq. \eqref{two-hops-b-moment} and Eq. \eqref{MBS-or-SBS-b-moment} (when $k=1$) in Eq. (\ref{total_meta_microwave}), we get the $b^{\mathrm{th}}$ moment  $M_{b,\text{T}}$. Finally, by substituting $M_{b,\text{T}}$ in Eq. (\ref{total_meta_microwave}) into either Eq. (\ref{Gil-pelaez-2hops}) or Eq. (\ref{beta-match}), we get the meta distribution of the $\mathrm{SIR}$. 
         %Finally, using Eq.  \eqref{MBS-or-SBS-b-moment} (when $k=1$) and Eq. \eqref{two-hops-b-moment}, we get the  $b^{\mathrm{th}}$ moment  $M_{b,\text{total}}$ in Eq. (\ref{total_meta_microwave}).

\subsection{Extensions to Millimeter-wave Backhauls Networks}\label{mmwave_backhaul_cor}
 The proposed framework can be extended to a scenario where the backhaul and access transmissions are conducted on orthogonal mm-wave spectrums. Note that Eq. (\ref{SIR-backhaul-O}) will be changed similar to Eq. (\ref{snr-access-O-mmwave}). Then, only the first term, $M_{b,\mathrm{BH}}(\theta_{2})$ in the main Eq. (\ref{total_meta}) of our model that characterizes the moment of the CSP in the backhaul will be re-defined as $M_{b,\mathrm{BH}}(\theta_{2})=\mathbb{E}[P^b_{s,2}(\theta_2)]$.

The framework can also be extended to a scenario where the backhaul transmissions are conducted on the mm-wave spectrum and the access links of SBSs operate on $\mu$-wave.
In this case, we will need to use the results in Section~VII.A while redefining the term $M_{b,\mathrm{BH}}(\theta_{2})$  as  $M_{b,\mathrm{BH}}(\theta_{2})=\mathbb{E}[P^b_{s,2}(\theta_2)]$ in \eqref{two-hops-b-moment}.

\section{Numerical Results and Discussions}\label{results_mm}
We present the simulation parameters in Section \ref{parameter_sec}. Then, we validate our numerical results using Monte-Carlo simulations in Section \ref{num_resu}. In Section \ref{num_resu}, we use the developed analytical models  to obtain insights related to the meta distribution of the SIR/SNR of a typical {\color{black}device}, mean and variance of the success probability, transmission delay, and the reliability of a typical {\color{black}device} in the downlink direction.

\subsection{Simulation Parameters}\label{parameter_sec}
Unless otherwise stated, we use the following simulation parameters throughout our numerical results. The transmission powers of MBSs and SBSs in the downlink are $P_{1}=50$ Watts and $P_{2}=5$ Watts, respectively. The size of the simulated network is $ 90\text{km}\times90 \text{km}$. We assume that the density of MBSs is $\lambda_{1}=2$ MBSs/km$^{2}$ and the density of SBSs is $\lambda_{2}=70$ SBSs/km$^{2}$. The offloading biases for the MBSs and the SBSs are  $B_{1}=B_{2}=1$, respectively. The PLE for MBSs is set to $\alpha_{1}=4$ and for mm-wave SBSs, $\alpha_{2,L}=2$ in the case of LOS and $\alpha_{2,N}=4$ in the case of NLOS. The network downlink bandwidth is 100 MHz for $\mu$wave MBSs and 1 GHz for mm-wave SBSs with channel frequency 28 GHz. The LOS (NLOS) states are modeled by large (small) values of $m$, i.e., {$m_L=2$ and $m_N=1$} \cite{turgut2017coverage}. SBSs number of antenna elements is $\mathcal{N}=10$. The receiver noise  is calculated as \cite{singh2015tractable}   $\sigma_{2}^{2}=-174\text{ dBm/Hz}+10\log_{10}(W_{2})+10\text{ dB}$, where $W_{2}=1\text{ GHz}$ is bandwidth allocated to the mm-wave SBSs. The antenna gains of MBSs are $G^{\mathrm{o}}_{1}=0$ dB and {\color{black}devices} directional antenna gain is $G^{\mathrm{max}}_{\mathcal{D}}=10$ dB.

%%%%%%%%%%%%%%%%%%%%%%%%%%%%%%%%%%%%%%%%%%%%%%%%%%%%%%%%%%%
\subsection{Numerical Results and Discussions}\label{num_resu}

\subsubsection{Association Probability}

Fig. \ref{association_total}  illustrates the accuracy of association probabilities in a hybrid spectrum network, derived in {\bf Lemma \ref{1}} and {\bf Lemma \ref{2}}, as a function of  $\lambda_2$ by showing a comparison  with Monte-Carlo simulations.  We notice from  Fig.~\ref{association_total} that by increasing the density of the
mm-wave SBSs $\lambda_2$, the probability of association with mm-wave LOS
SBSs $\mathcal{A}_{2,L}$ increases which confirms the insights from {\bf Corollary~1} and {\bf Corollary~2}. The reason is the increasing number of SBSs per unit area within the LOS ball will favour the {\color{black}device} association towards LOS SBSs and reduces the chances of associating with NLOS SBSs.  The addition of
$\mathcal{A}_{1}+\mathcal{A}_{2,L}+\mathcal{A}_{2,N}=1$ is equal  to
unity for different densities of SBSs  $\lambda_2$. Note that the probability of associating with $\mu$wave MBSs is minimal due to a higher path-loss exponent and NLOS omnidirectional transmissions from MBSs.

\subsubsection{The Meta Distribution of the {\color{black} SIR/}SNR}
In Fig. \ref{meta_total_mmwave}, we validate our analytical results for the meta distribution of the {\color{black}SIR/}SNR of a typical {\color{black}device} in a hybrid spectrum network through simulations.  Fig. \ref{meta_total_mmwave} also  depicts the probability of achieving reliability $x$, i.e., $x\%$ fraction of {\color{black}devices} can achieve their quality of service for $\theta \in \{10, 1, 0.1\}$ dB. From Fig. \ref{meta_total_mmwave}, we note that about 18\% of the devices (when $\theta=10$), 51\% of devices (when $\theta=1$), and 96\% of devices (when $\theta=0.1$) have success probabilities equal to $0.3$.

\begin{figure}[ht]
\begin{minipage}[b]{.48\textwidth}
\centering
\includegraphics[width=0.95\textwidth]{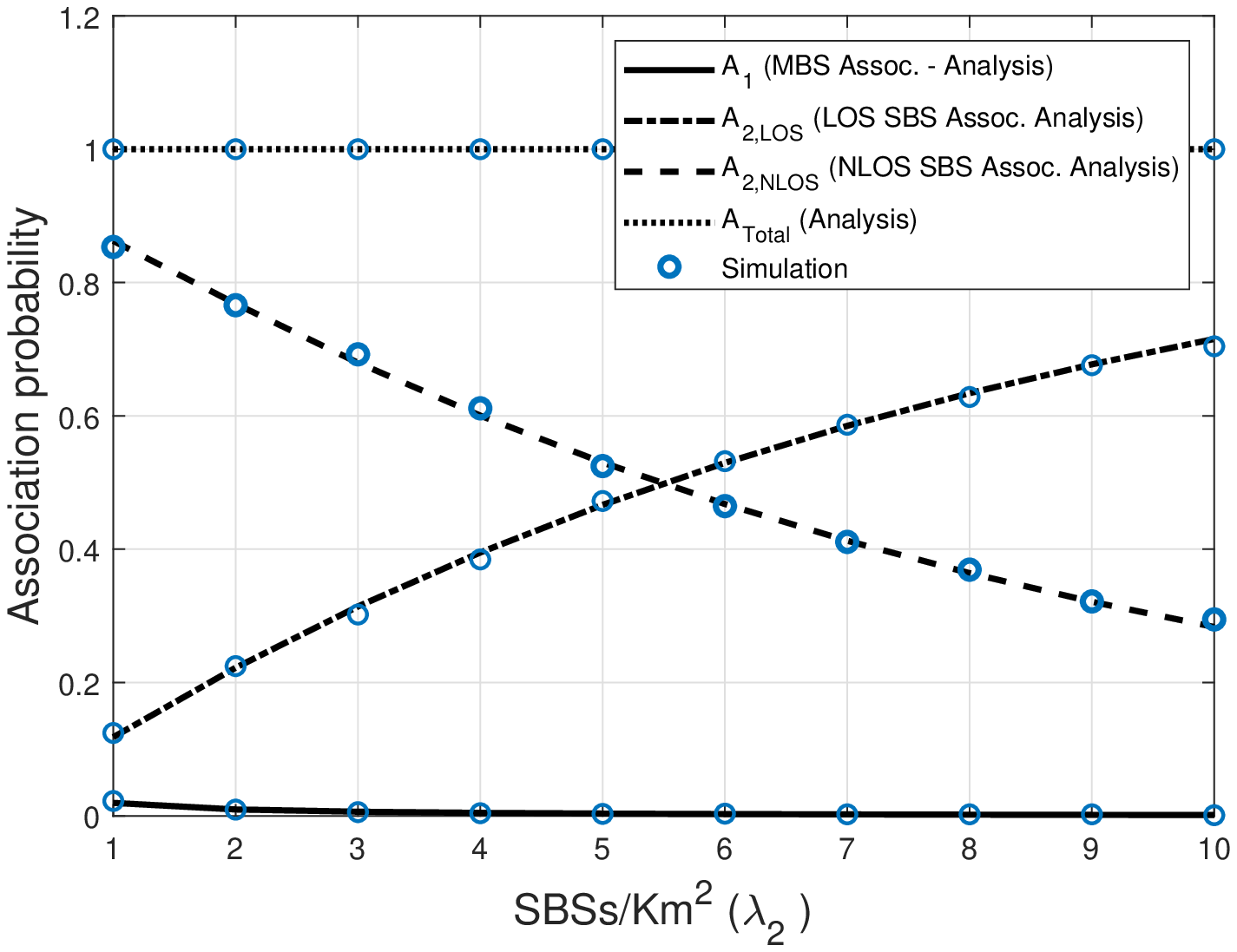}
\caption{\small Association probabilities as a function of $\lambda_{2}$ for the hybrid spectrum network when $\lambda_1=2\text{ MBSs/km}^{2}$, $B_1=B_2=1$,  and $d=200\text{m}$.}
\label{association_total}
\end{minipage}
\hfill
\begin{minipage}[b]{.48\textwidth}
\centering
\includegraphics[width=0.9\textwidth]{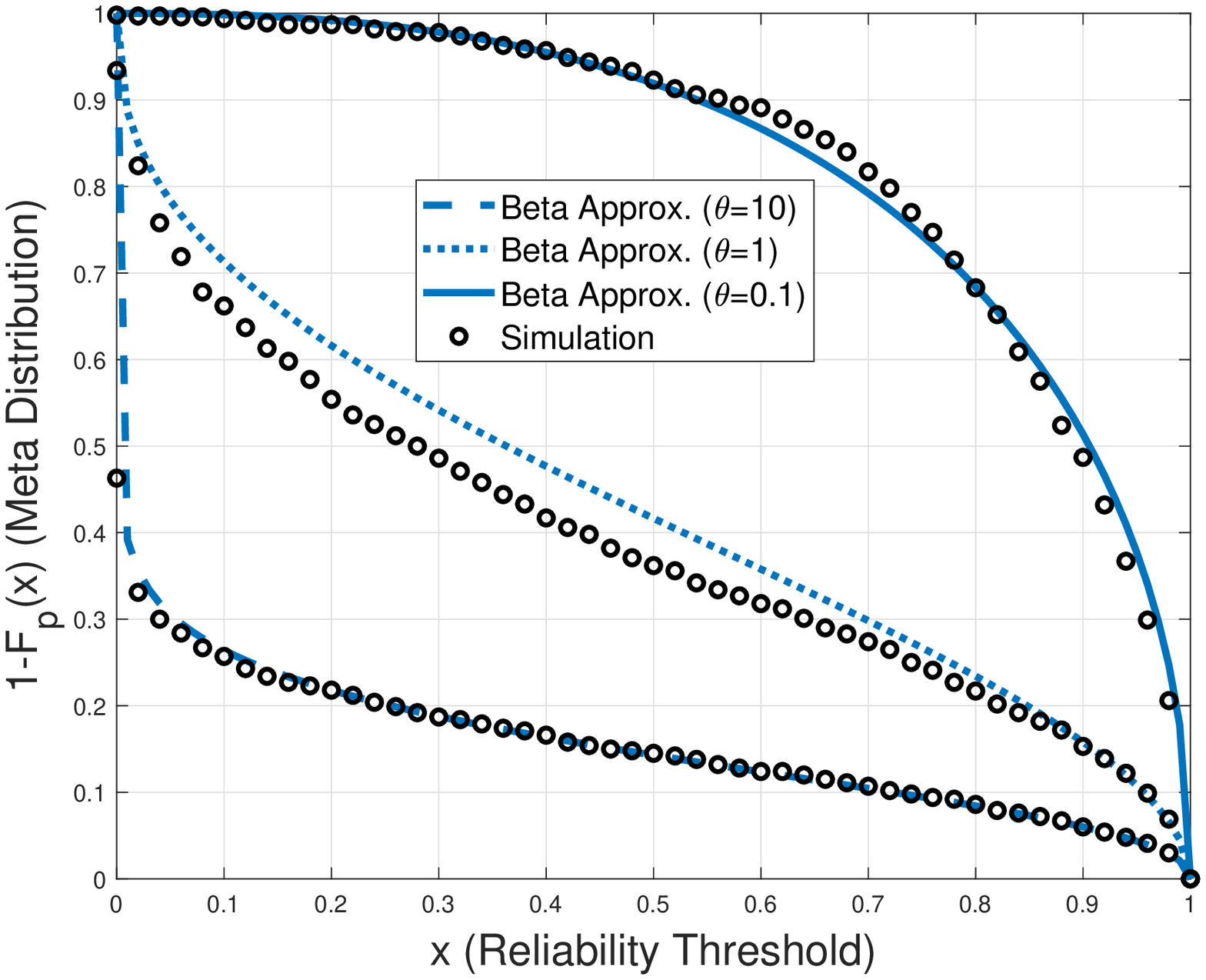}
\caption{\small The meta distribution vs. reliability threshold  $x$ for $\theta=\theta_{\mathcal{D}}=\theta_{2}=$10, 1, and 0.1 for the hybrid spectrum network when $B_{1}=B_{2}=1$,  and $d=200\text{m}$.}
\label{meta_total_mmwave}
\end{minipage}
\end{figure}

%\begin{figure}[!h]
%\begin{center}
%\scalebox{0.43}[0.43]{\includegraph%ics{Association_Probv1.eps}}
%\vspace{-.5cm}
%\end{center}
%\caption{\small Association %probabilities as a function of %$\lambda_{2}$ for the hybrid %spectrum {\color{black}IoT} network %when $\lambda_1=2\text{ %MBSs/km}^{2}$, $B_1=B_2=1$,  and %$d=200\text{m}$.}
%\vspace{-.2cm}
%\label{association_total}
%\end{figure}
%\begin{figure}[!h]
%\begin{center}
%\scalebox{0.35}[0.35]{\includegrap%hics{meta.eps}}
%\vspace{-.5cm}
%\end{center}
%\caption{\small The meta distribution vs. reliability threshold  $x$ for $\theta=\theta_{\mathcal{D}}=\theta_{2}=$10, 1, and 0.1 for the hybrid spectrum {\color{black}IoT} network when $B_{1}=B_{2}=1$,  and $d=200\text{m}$.}
%\vspace{-.2cm}
%\label{meta_total_mmwave}
%\end{figure}
\subsubsection{Coverage and Variance as a Function of SIR/SNR Threshold in Hybrid Spectrum Networks}
Fig. \ref{coverage_varianace_mmwave} illustrates the standard success probability  $M_{1,\text{T}}$ and its variance $M_{2,\text{T}}-M^2_{1,\text{T}}$ as a function of target {\color{black}SIR/}SNR threshold $\theta$ of {\color{black}devices} in a hybrid spectrum network. As we can see  in Fig. \ref{coverage_varianace_mmwave} that the simulation results match the analytical results, however the slight gap is due to the Alzer's inequality considered in {\bf Appendix~C}. This gap will be zero when Nakagami fading turns into Rayleigh fading as shown in the next figure. 
%Since the variance must equal to zero for both $\theta \rightarrow 0$ and $\theta \rightarrow \infty$, it goes to a maximum at some finite value of $\theta$. 
By examining Fig. \ref{coverage_varianace_mmwave}, a numerical evaluation shows that the variance is maximized at $\theta=-3$ dB where the success is $M_{1,\text{T}} = 0.49$. For moderate values of $\theta$, there is a trade-off between maximizing coverage or reducing variance because the variance first increases and then decreases while the coverage probability is monotonically decreasing. For higher values of $\theta$,  lower coverage probabilities have lower variance so its a low-reliability regime where more {\color{black}devices}' performances are spread around low coverage probability. As such, the low values of $\theta$ provides a higher reliability regime.

Fig. \ref{coverage_varianace_rayleigh} illustrates the standard
success probability $M_{1,\text{T}}$ and the variance as
a function of $\theta$ with Rayleigh fading (i.e., $m_L=m_N=1$). As we can see  in   Fig. \ref{coverage_varianace_rayleigh} that the simulation results closely match the analytical results. The reason is that the approximation of the incomplete Gamma function (also referred to as Alzer's inequality) becomes exact when $m_L$ becomes equal to unity.
Subsequently, this figure explains the reason for the gap between the simulation and the analytical curves in Fig. \ref{coverage_varianace_mmwave}.

\begin{figure}[ht]
\begin{minipage}[b]{.48\textwidth}
\centering
\includegraphics[width=0.85\textwidth]{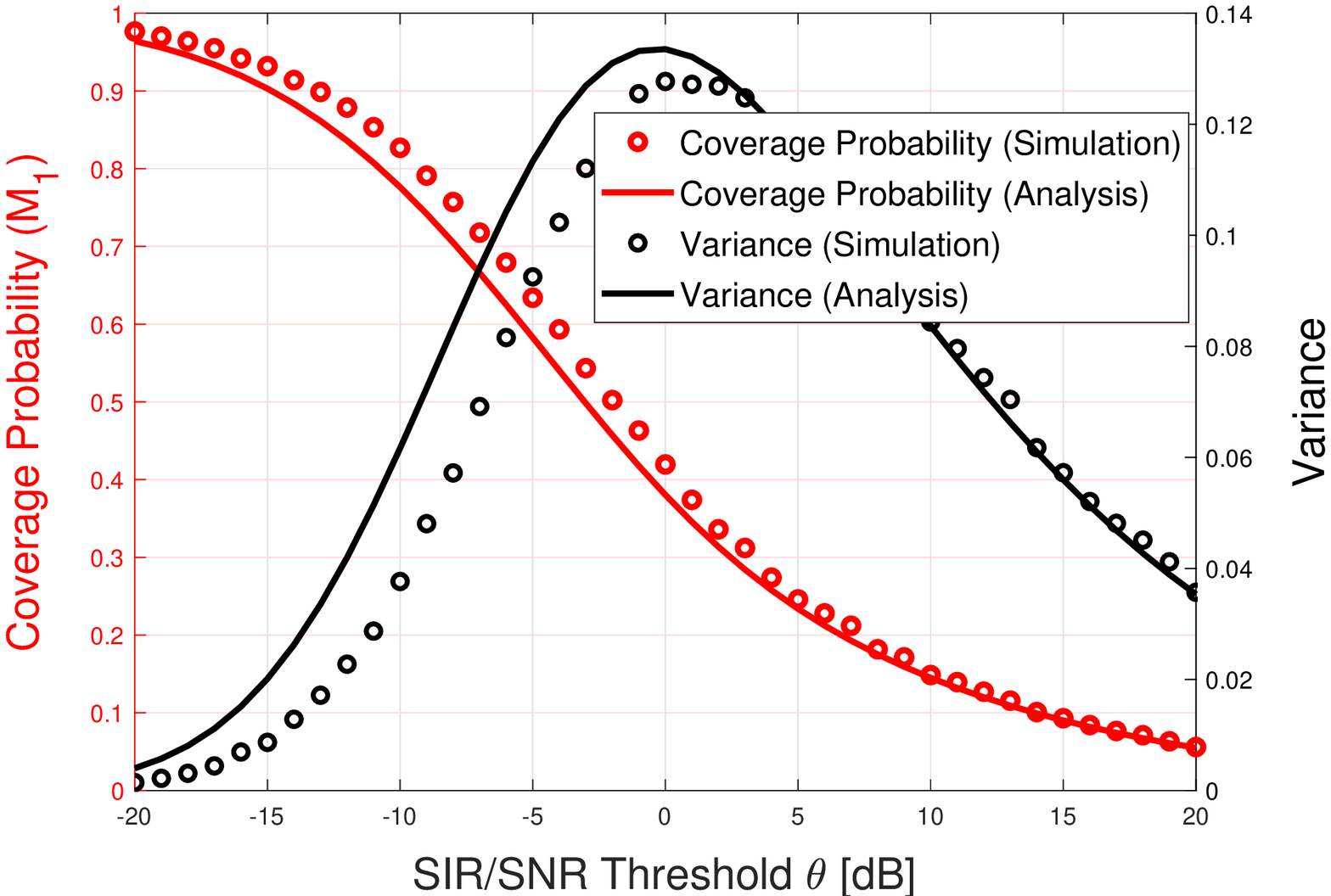}
\caption{\small Coverage probability $M_{1,\text{T}}$ and  variance $M_{2,\text{T}}-M_{1,\text{T}}^{2}$ as a function of $\theta$ considering Nakagami-m fading when $B_{1}=B_{2}=1$ and $d=200\text{m}$.}
\label{coverage_varianace_mmwave}
\end{minipage}
\hfill
\begin{minipage}[b]{.48\textwidth}
\centering
\includegraphics[width=0.9\textwidth]{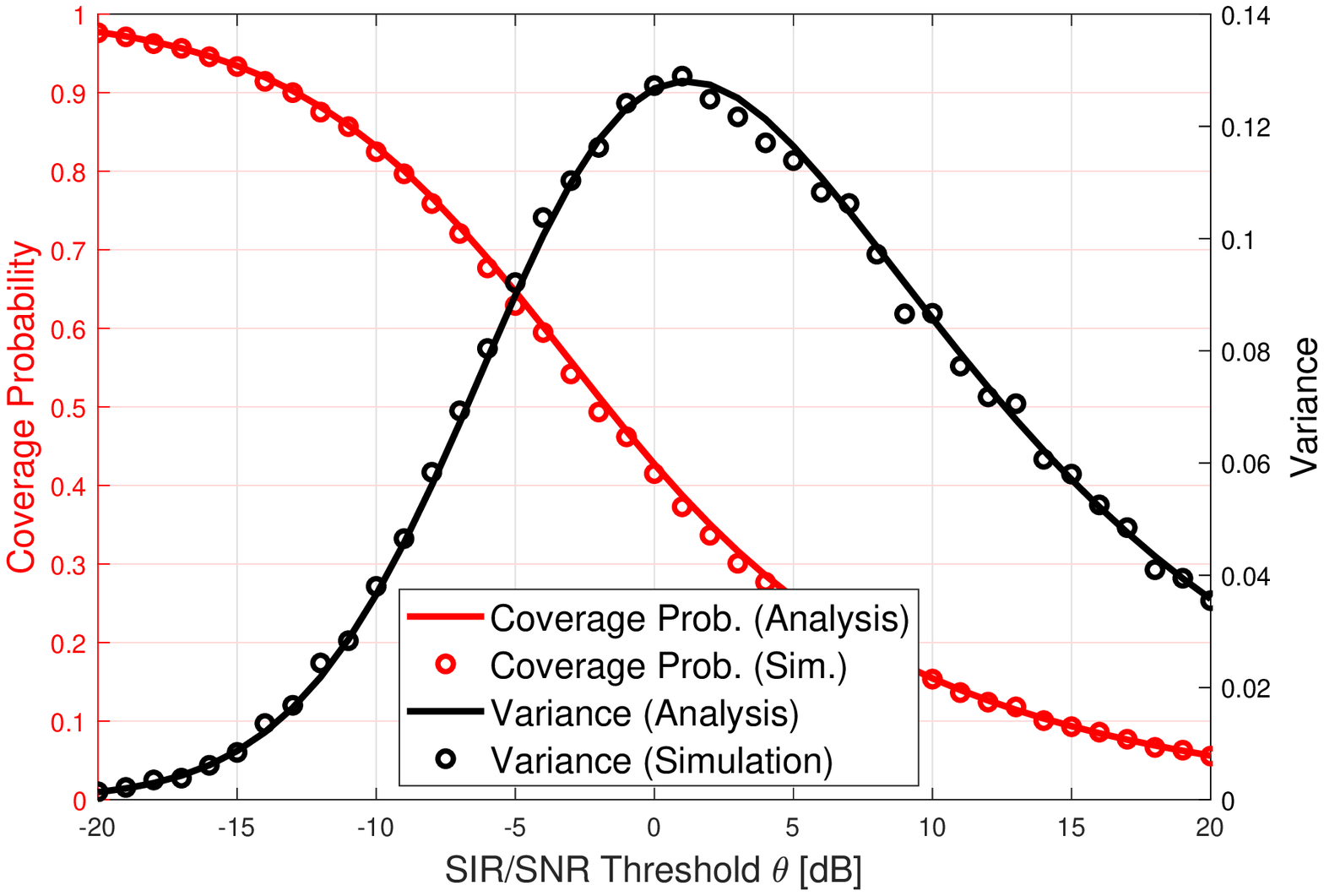}
\caption{\small Coverage probability $M_{1,\text{T}}$ and variance	$M_{2,\text{T}}-M_{1,\text{T}}^{2}$ as a function of $\theta$ considering Rayleigh fading (i.e., $m_L=m_N=1$, when $B_{1}=B_{2}=1$  and $d=200\text{m}$.}
\label{coverage_varianace_rayleigh}
\end{minipage}
\end{figure}

% \begin{figure}[!h]
%\begin{center}
%\scalebox{0.35}[0.35]{\includegraphics{coverage_2.eps}}
%\vspace{-.5cm}
%\end{center}
%\caption{\small Coverage probability $M_{1,\text{T}}$ and  variance $M_{2,\text{T}}-M_{1,\text{T}}^{2}$ as a function of $\theta$ considering Nakagami-m fading when $B_{1}=B_{2}=1$ and $d=200\text{m}$.}
%\vspace{-.2cm}
%\label{coverage_varianace_rayleigh}
%\end{figure}

%\begin{figure}[!h]
%\begin{center}
%\scalebox{0.35}[0.35]{\includegraphics{coverage_rayleigh.eps}}
%\vspace{-.5cm}
%\end{center}
%\caption{\small Coverage probability $M_{1,\text{T}}$ and variance	$M_{2,\text{T}}-M_{1,\text{T}}^{2}$ as a function of $\theta$considering Rayleigh fading (i.e.,	$m_L=m_N=1$, when $B_{1}=B_{2}=1$  and $d=200\text{m}$.}
%\vspace{-.2cm}
%\label{coverage_varianace_rayleigh}
%\end{figure}

 \subsubsection{Coverage and Variance as a Function of  the Number of Antenna Array Elements in Hybrid Spectrum Networks}
 Fig. \ref{coverage_varianace_antenna} depicts the coverage probability and variance as a function of $\theta$ considering the number of antenna array elements  as $\mathcal{N}=10$, 20, and 30 to show the effect of higher directional antenna gains. The general trends for the coverage probability and its variance are found to be the same as in previous figures. The main observation is that although the coverage enhancement is not significant with increasing antenna elements,  the reduction in the variance is noticeable which supports  higher directional antenna gains and the importance of analyzing the higher moments of the CSP.

 \subsubsection{Coverage and Variance as a Function of $B_{2}$ in $\mu$wave-only Networks}
In Fig. \ref{coverage_variance_bias}, we study the effect of offloading {\color{black}devices} from the MBS tier to the SBSs tier in terms of  the coverage probability (which is the mean reliability) and the variance of the CSP (or reliability). By offloading {\color{black}devices} from the MBS tier to the SBSs tier when $B_2=30$, the coverage probability $M_{1,\text{T}}$ suffers due to the dual-hop transmission effect in wireless backhauled SBSs; however the variance of the results reduces which is a novel and positive insight. Another observation is that the variance of the CSP in $\mu$wave-only network is high compared to the hybrid network. This can be shown by comparing points $V_{1}=(1,0.1)$ in Fig. \ref{coverage_varianace_antenna} and $V_{2}=(4,0.19)$ in Fig. \ref{coverage_variance_bias}, for the case of $B_1=B_2=1$. We noticed that the variance has decreased from 0.19 to 0.1 when the SBS antenna array size is increased to $\mathcal{N}=20$. This implies that the hybrid spectrum network outperforms the $\mu$wave-only network due to the directional antenna gains.

\begin{figure}[ht]
\begin{minipage}[b]{.48\textwidth}
\centering
\includegraphics[width=0.85\textwidth]{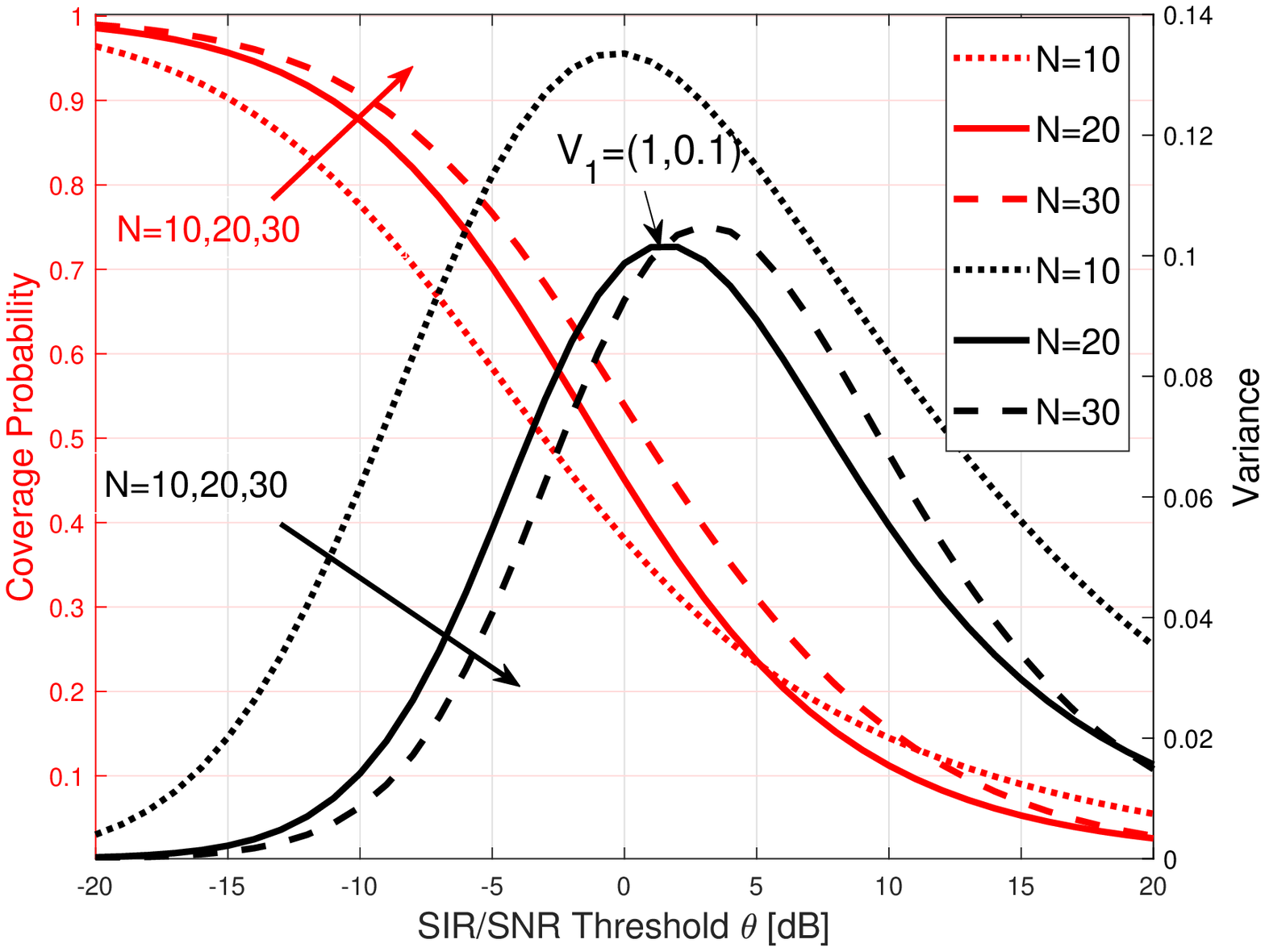}
\caption{\small Coverage probability $M_{1,\text{T}}$ and  variance $M_{2,\text{T}}-M_{1,\text{T}}^{2}$ as a function of {\color{black} $\mathcal{N}$} for hybrid spectrum network when $B_{1}=B_{2}=1$, and $d=200\text{m}$.}
\label{coverage_varianace_antenna}
\end{minipage}
\hfill
\begin{minipage}[b]{.48\textwidth}
\centering
\includegraphics[width=0.95\textwidth]{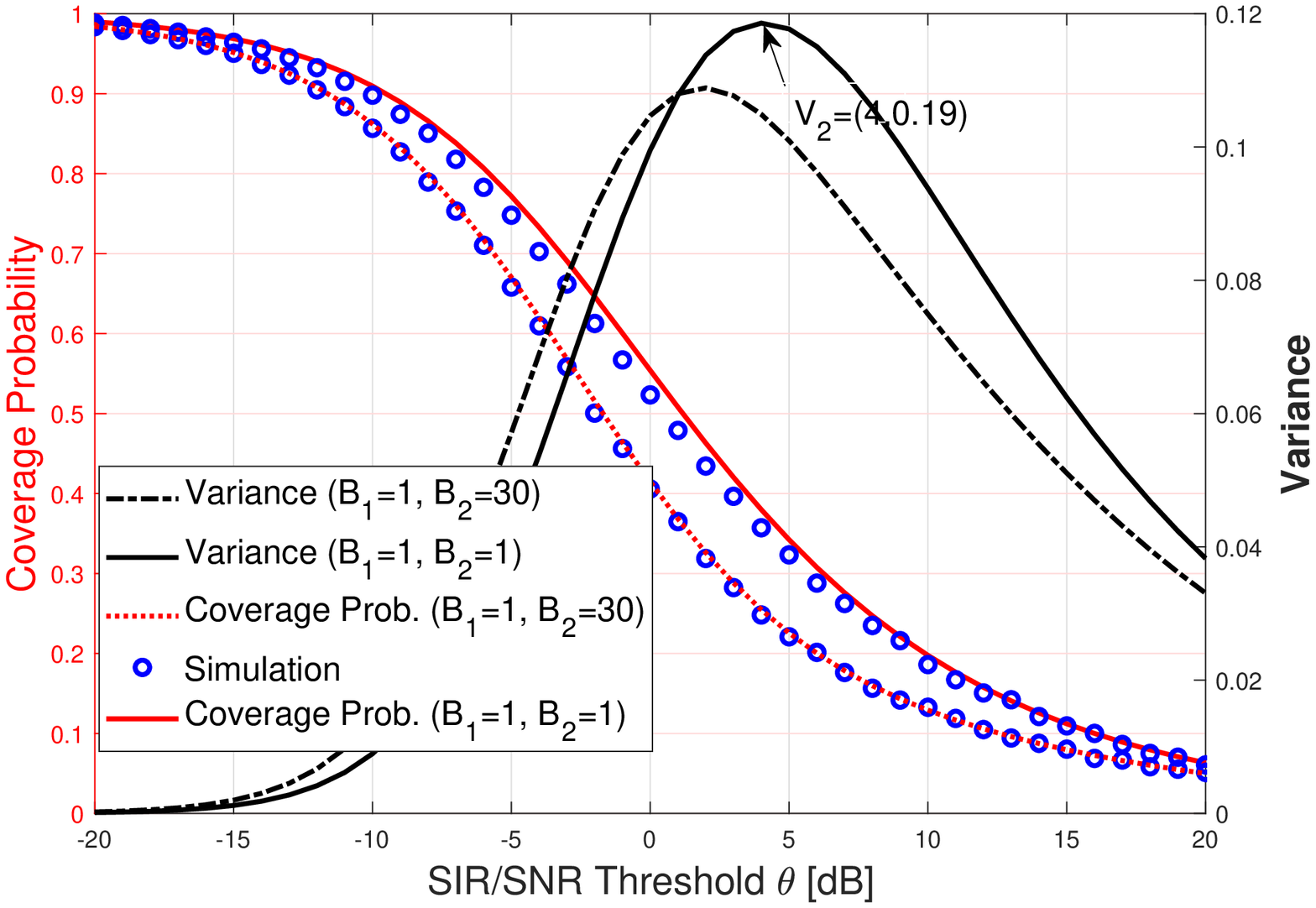}
\caption{\small Coverage probability $M_{1,\text{T}}$ and variance $M_{2,\text{T}}-M_{1,\text{T}}^{2}$ as a function of $\theta$ for $\mu$wave-only network when $\alpha_{1}=\alpha_{2}=4$, $B_{1}=1$, and $B_{2}=1$ and 30.}
\label{coverage_variance_bias}
\end{minipage}
\end{figure}

%\begin{figure}[!h]
%\begin{center}
%\scalebox{0.35}[0.35]{\includegraphics{antenna_coverage_2.eps}}
%\vspace{-.5cm}
%\end{center}
%\caption{\small Coverage probability $M_{1,\text{T}}$ and  variance $M_{2,\text{T}}-M_{1,\text{T}}^{2}$ as a function of {\color{black} $\mathcal{N}$} for hybrid spectrum {\color{black}IoT} network when $B_{1}=B_{2}=1$, and $d=200\text{m}$.}
%\vspace{-.2cm}
%\label{coverage_varianace_antenna}
%\end{figure}

%\begin{figure}[!h]
%\begin{center}
%\scalebox{0.35}[0.35]{\includegraphics{coverage_beta_selfbackhaul_orthog_bias.eps}}
%\vspace{-.5cm}
%\end{center}
%\caption{\small Coverage probability $M_{1,\text{T}}$ and variance $M_{2,\text{T}}-M_{1,\text{T}}^{2}$ as a function of $\theta$ for $\mu$wave-only {\color{black}IoT} network when $\alpha_{1}=\alpha_{2}=4$, $B_{1}=1$, and $B_{2}=1$ and 30.}
%\vspace{-.2cm}
%\label{coverage_variance_bias}
%\end{figure}
 \subsubsection{Mean Local Delay ($\mu$wave vs mm-wave SBSs)}
 Fig. \ref{local_delay_mmwave} depicts the mean local delay experienced by a typical {\color{black}device} as a function of the SBSs density $\lambda_2$ in a hybrid spectrum network. The mean local delay is the mean number of transmission attempts  to successfully transmit a packet.
 The mean local delay
 increases by increasing $\lambda_2$. After  the SBS density reaches $\lambda_2=20\text{
SBSs/km}^{2}$, the mean local delay stays constant at value 1.11. This result can be intuitively explained as follows. When the mm-wave SBS density is low, the typical {\color{black}device} has a higher probability to connect to a MBS, i.e., the mean local delay of the network results from only one hop communication (from the MBS to the {\color{black}device}). However, when the $\lambda_2$ increases, the typical {\color{black}device} has a higher probability to connect to a mm-wave SBS, i.e., the network local delay results from two hops communication (from the MBS to the SBS then from the SBS to the {\color{black}device}). Furthermore, the beamforming high directional gain steerable antennas will push more {\color{black}devices} to associate with SBSs thus a higher network delay is observed.  Fig. \ref{local_delay} shows that, all else being equal, the mean local delay of the hybrid spectrum network is lower than that of the $\mu$wave-only network.

\begin{figure}[ht]
\begin{minipage}[b]{.48\textwidth}
\centering
\includegraphics[width=0.91\textwidth]{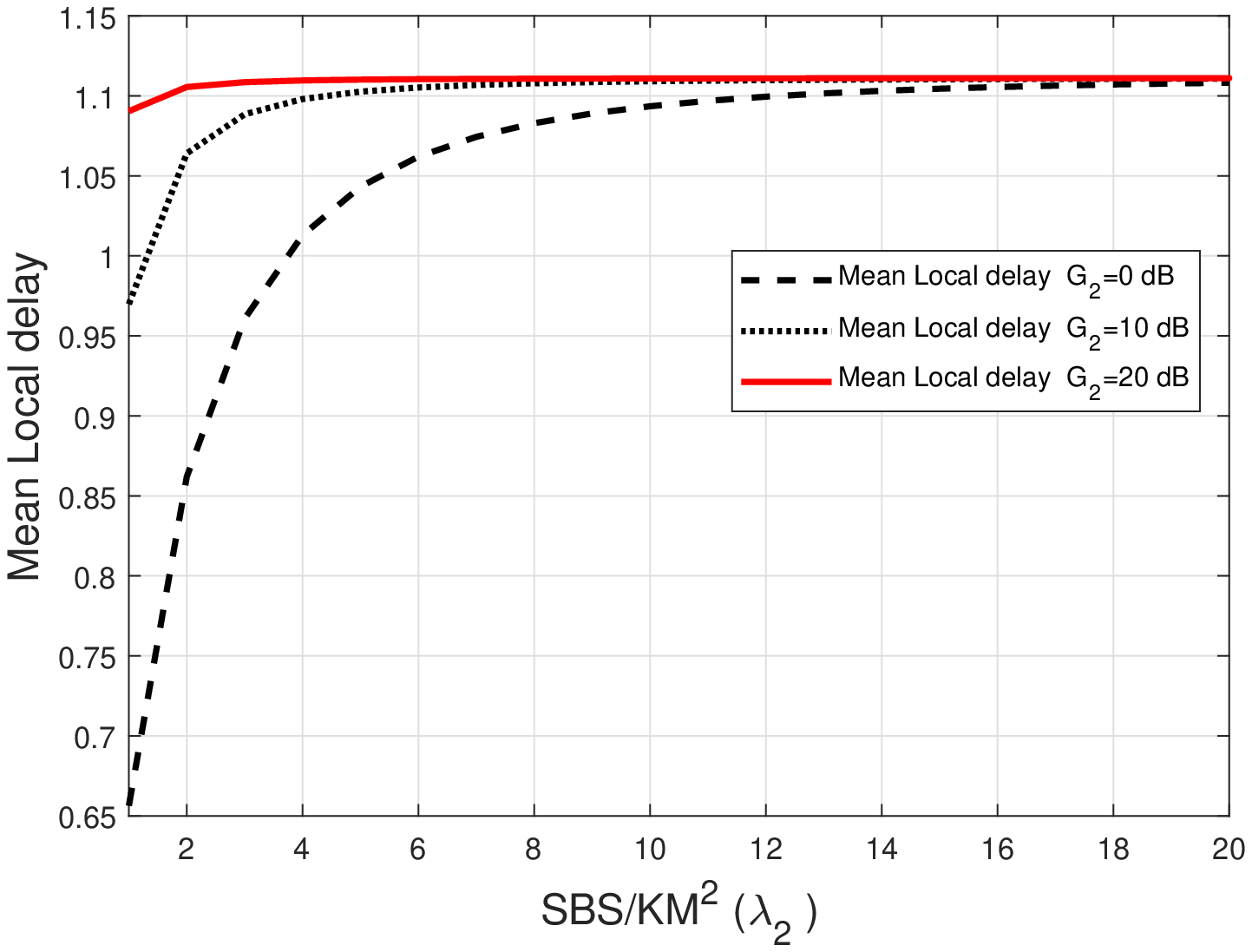}
\caption{\small Mean local delay $M_{-1,\text{T}}$ as a function of $\lambda_2$  for the hybrid spectrum network  when $\lambda_1=2\text{ MBS/Km}^{2}$, $B_{1}=1$, $B_{2}=10$, $\alpha_{1}=4$, $d=200\text{m}$, and $\theta=\theta_{\mathcal{D}}=\theta_{2}=$-10 dB.}
\label{local_delay_mmwave}
\end{minipage}
\hfill
\begin{minipage}[b]{.48\textwidth}
\centering
\includegraphics[width=1.03\textwidth]{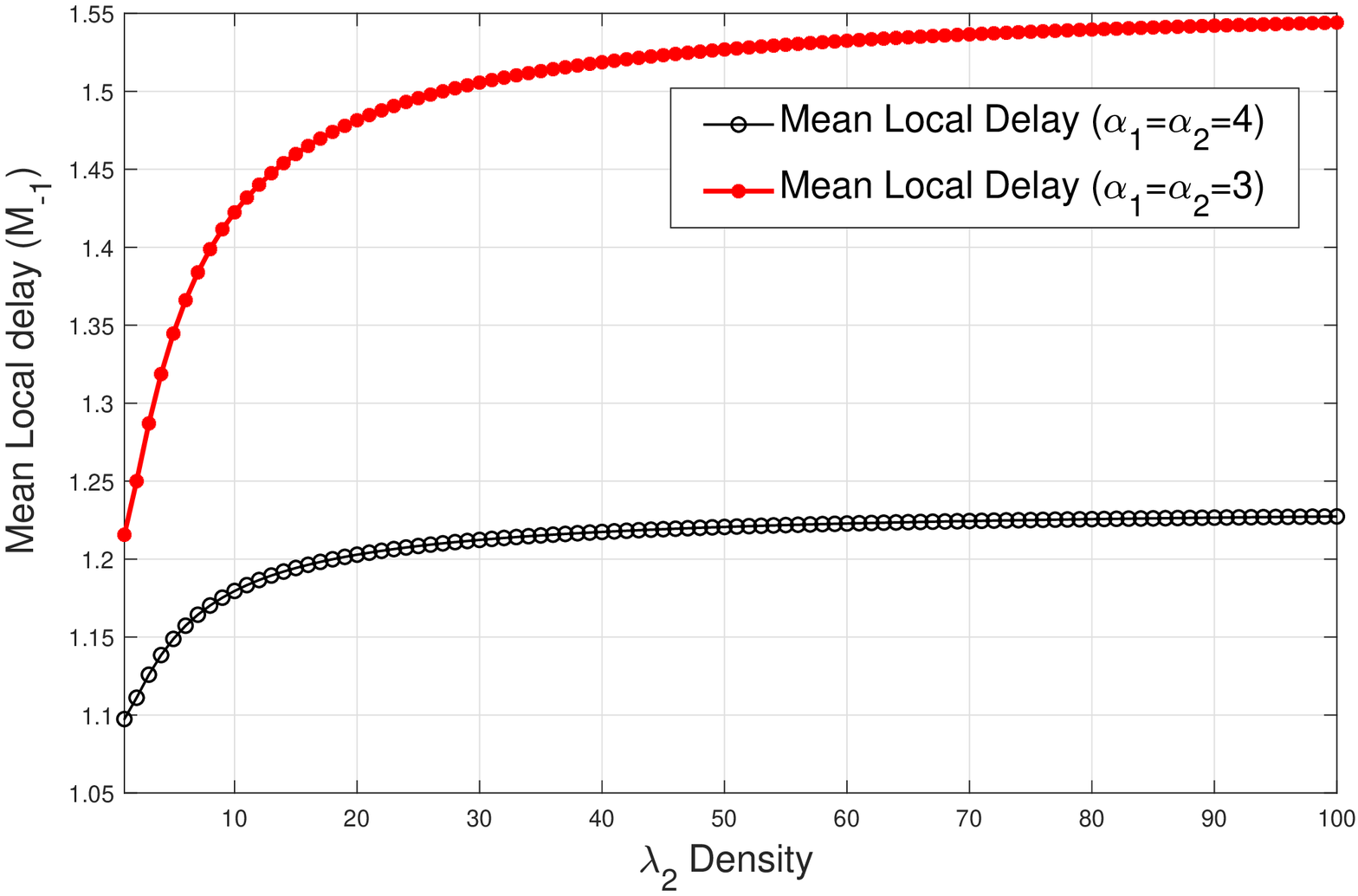}
\caption{\small Mean local delay $M_{-1,\text{T}}$ as a function of $\lambda_2$ for the $\mu$wave-only network when $\lambda_1=2\text{ MBS/KM}^{2}$, $B_{1}=1$ and $B_{2}=10$, $\alpha_{1}=\alpha_{2}=3$ and 4, and $\theta=\theta_{\mathcal{D}}=\theta_{2}=$-10 dB.}
\label{local_delay}
\end{minipage}
\end{figure}

%\begin{figure}[!h]
%\begin{center}
%\scalebox{0.5}[0.5]{\includegraphics{Local_delay_beta_selfbackhaul_orthog_density_lambda_2_antena_gain_3.eps}}
%\vspace{-.5cm}
%\end{center}
%\caption{\small Mean local delay $M_{-1,\text{T}}$ as a function of $\lambda_2$  for the hybrid spectrum {\color{black}IoT} network  when $\lambda_1=2\text{ MBS/Km}^{2}$, $B_{1}=1$, $B_{2}=10$, $\alpha_{1}=4$, $d=200\text{m}$, and $\theta=\theta_{\mathcal{D}}=\theta_{2}=$-10 dB.}
%\vspace{-.2cm}
%\label{local_delay_mmwave}
%\end{figure}

%\begin{figure}[!h]
%\begin{center}
%\scalebox{0.35}[0.35]{\includegraphics{Local_delay_beta_selfbackhaul_orthog_density_lambda_2.eps}}
%\vspace{-.5cm}
%\end{center}
%\caption{\small Mean local delay $M_{-1,\text{T}}$ as a function of $\lambda_2$ for the $\mu$wave-only {\color{black}IoT}  network when $\lambda_1=2\text{ MBS/KM}^{2}$, $B_{1}=1$ and $B_{2}=10$, $\alpha_{1}=\alpha_{2}=3$ and 4, and $\theta=\theta_{\mathcal{D}}=\theta_{2}=$-10 dB.}
%\vspace{-.2cm}
%\label{local_delay}
%\end{figure}

Fig. \ref{local_delay} depicts the mean local delay for a $\mu$wave-only network as a function of $\lambda_2$. When $\lambda_2$ increases the mean local delay of the total network increases again due to the increase in interference which is not the case in the hybrid spectrum  network. The network mean local delay in the case of $\alpha_{1}=\alpha_{2}=3$ is higher than that in the case of $\alpha_{1}=\alpha_{2}=4$ due to higher path loss degradation for higher PLEs.
\subsubsection{The Meta Distribution of the Achievable Data Rate in Hybrid Spectrum Networks}
Fig. \ref{meta_data_rate} depicts the meta distribution of the data rate in hybrid spectrum networks as a function of reliability $x$ for different number of antenna elements $\mathcal{N}=10$, 20, 40, and 50 with rate threshold $\mathcal{T}=1$ Gbps. As shown in Fig. \ref{meta_data_rate}, the fraction of devices achieving a required rate increases as the number of antennas elements increases. In other words, increasing the number of antenna elements of SBSs has a positive effect on the achievable rate and its meta distribution. This insight helps 5G cellular network operators to find the most efficient operating antenna configuration to achieve certain reliability for certain 5G applications.

\subsubsection{The Meta Distribution in a Microwave-only  Network}
In Fig. \ref{meta_selfbackhaul}, we validate our analysis by depicting the exact (Gil-Pelaez) meta distribution in a $\mu$wave-only network defined in Eq. (\ref{Gil-pelaez-2hops}), and the beta approximation for the meta distribution defined in Eq. (\ref{beta-match}). Our simulation result  provides an excellent match for a wide range of $\theta$ values and this validates the correctness of our analytical model. Fig. \ref{meta_selfbackhaul} also serves as an illustration of the meta distribution of the SIR of a typical {\color{black}device} in a $\mu$wave-only network. We note that about 23\% of {\color{black}devices} (when $\theta=10$), 72\% of {\color{black}devices} (when $\theta=1$), and 98\% of {\color{black}devices} (when $\theta=0.1$) have reliability, i.e., success probability, equal to $0.3$.

\begin{figure}[ht]
\begin{minipage}[b]{.48\textwidth}
\centering
\includegraphics[width=0.8\textwidth]{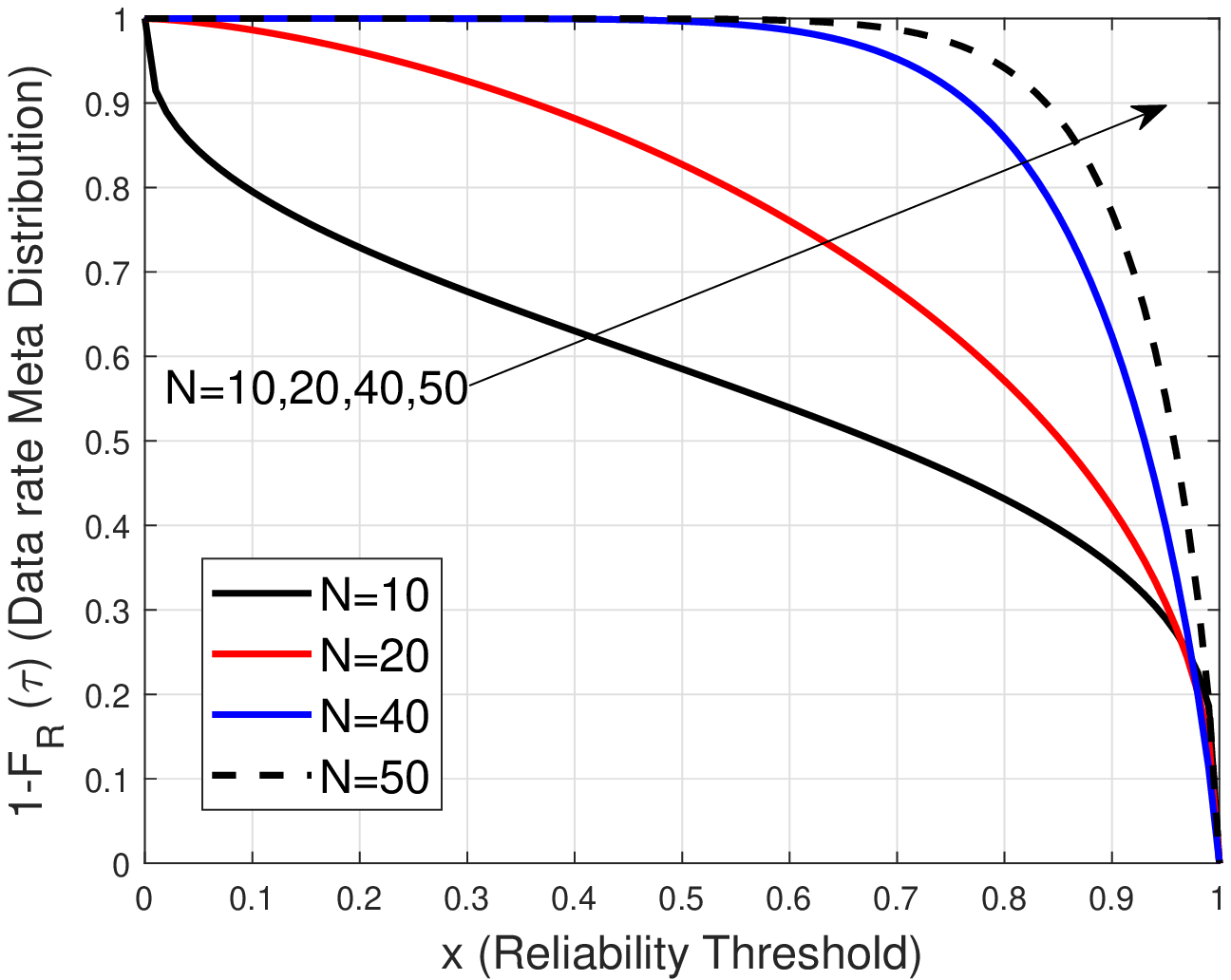}
\caption{\small Meta distribution of the achievable data rate as a function of reliability $x$ for different number of antenna elements $\mathcal{N}$ with rate threshold $\mathcal{T}=1$ Gbps.}
\label{meta_data_rate}
\end{minipage}
\hfill
\begin{minipage}[b]{.48\textwidth}
\centering
\includegraphics[width=0.9\textwidth]{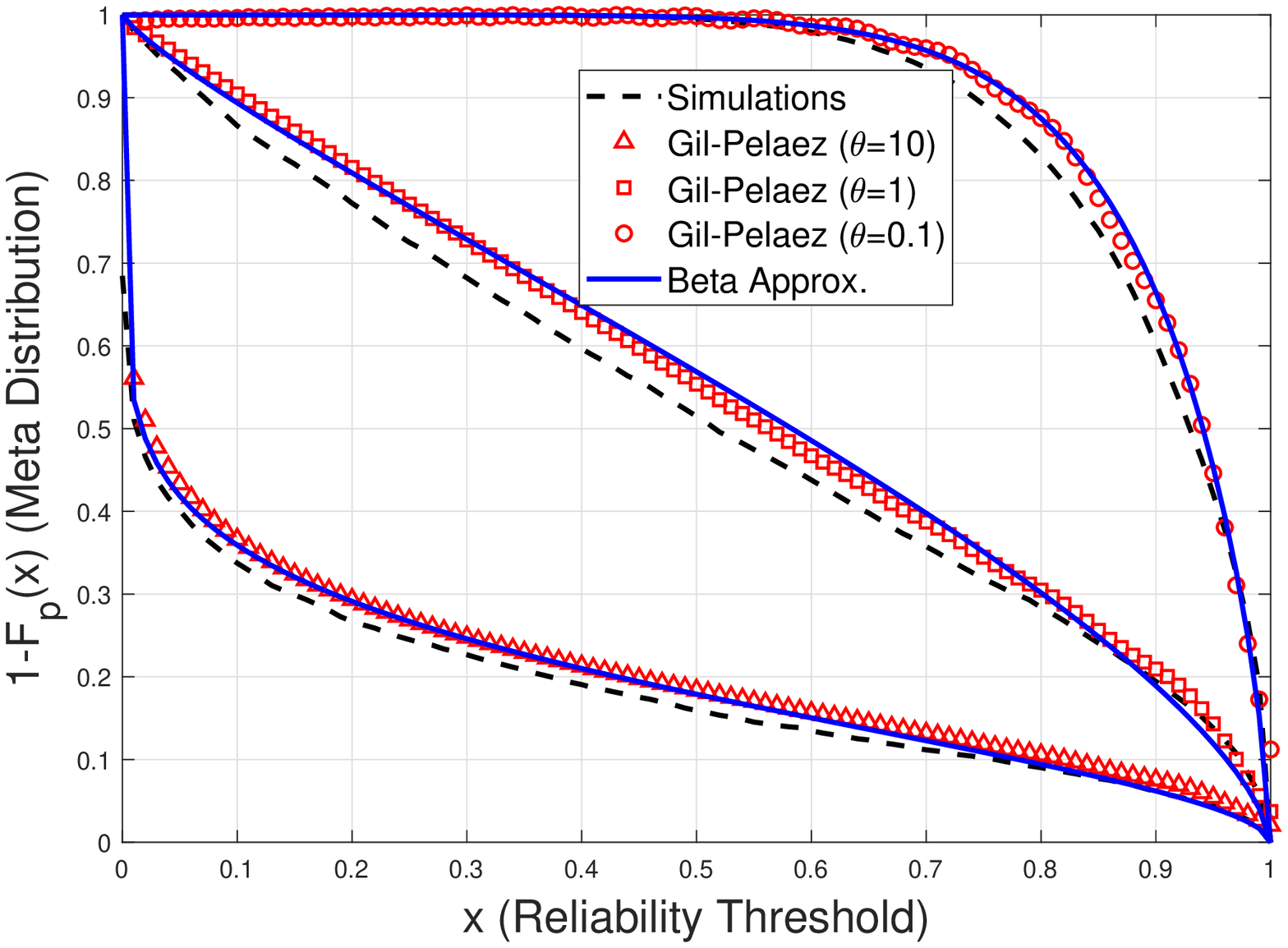}
\caption{\small The meta distribution as a function of reliability $x$ for $\theta=\theta_{\mathcal{D}}=\theta_{2}=$10, 1, and 0.1 for SBSs in a $\mu$wave-only network when $B_{1}=B_{2}=1$ and $\alpha_{1}=\alpha_{2}=4$.}
\label{meta_selfbackhaul}
\end{minipage}
\end{figure}

%\begin{figure}[!h]
%\begin{center}
%\scalebox{0.45}[0.45]{\includegraphics{Data_rate_meta_vs_N.eps}}
%\vspace{-.5cm}
%\end{center}
%\caption{\small Meta distribution of the achievable data rate as a function of reliability $x$ for different number of antenna elements $\mathcal{N}$ with rate threshold $\mathcal{T}=1$ Gbps.}
%\vspace{-.2cm}
%\label{meta_data_rate}
%\end{figure}

%\begin{figure}[!h]
%\begin{center}
%\scalebox{0.35}[0.35]{\includegraphics{meta-selfbackhaulv1.eps}}
%\vspace{-.5cm}
%\end{center}
%\caption{\small The meta distribution as a function of reliability $x$ for $\theta=\theta_{\mathcal{D}}=\theta_{2}=$10, 1, and 0.1 for SBSs in a $\mu$wave-only {\color{black}IoT} network when $B_{1}=B_{2}=1$ and $\alpha_{1}=\alpha_{2}=4$.}
%\vspace{-.2cm}
%\label{meta_selfbackhaul}
%\end{figure}

\section{Conclusion}\label{conc_future_mm}
This paper characterizes the meta distributions of the {\color{black} SIR/SNR} and data rate of a typical {\color{black}device} in a hybrid spectrum network and $\mu$wave-only network.  The meta distribution is evaluated first by formulating and then characterizing the moments of the CSP of a typical {\color{black}device} in the hybrid network. Important performance metrics such as the mean local delay, coverage probability, network jitter, and variance of the CSP (or reliability) are studied. Numerical results demonstrate the significance of evaluating the meta distribution which requires a systematic evaluation of the generalized moment of order $b$ that helps in evaluating network metric such as coverage probability when $b=1$, mean local delay when $b=-1$, network jitter using $b=-2$ and $b=-1$, etc. Numerical results provide valuable insights related to the reliability and latency of the hybrid spectrum network and $\mu$wave-only network. These insights would help 5G cellular network operators to find the most efficient operating antenna configuration  to  achieve  certain  reliability  for  certain  5G applications.
%the effect of increasing the SBS density on the performance of a typical user in terms of the mean local delay and the meta distribution. For the  SBSs in a $\mu$wave-only network, the mean local delay increases as the SBS density increases. However, in the case of hybrid spectrum networks, the mean local delay stays constant, as the SBS density increases, thanks to a higher {\color{black}number of antenna array elements} in the SBSs.  Although both types of networks have similar coverage probabilities, the variance in a hybrid spectrum network is lower than that in a microwave-only network. This kind of fine-grained performance evaluation can only be achieved by using the meta distribution rather than the traditional coverage probability.

\section*{Appendix C: Proof of Lemma \ref{3}}\label{prof_3}
The $b^{\mathrm{th}}$ moment of the CSP of a typical {\color{black}device} served by the mm-wave SBS can be derived as:
\small
\begingroup \allowdisplaybreaks
      \begin{align}\label{MBS-or-SBS-b-moment-mmwave_2_new}
            &M_{b,\text{2}}(\theta_{\mathcal{D}})=\mathbb{E}_{l}\bigg[\underbrace{\mathbb{P}(n=2|L_{2,min}=l_2)}_{\bar{\mathcal{A}}_{2}(l_2)}P_{s,2}(\theta_{\mathcal{D}})^{b}\bigg]=
        \mathbb{E}_{l}\bigg[\bar{\mathcal{A}}_{2}(l_2)\left(p_{L} P_{s,2,L}(\theta_{\mathcal{D}})+p_{N}P_{s,2,N}(\theta_{\mathcal{D}})\right)^{b}\bigg],
            %%%%
            \notag\\
                 &\stackrel{(a)}{=}\mathbb{E}_{l}\Bigg[\bar{\mathcal{A}}_{2}(l_2)\left(p_{L} \frac{\Gamma\left(m_{L},\frac{m_{L}}{\Omega_{L}}\nu_{L}\right)}{\Gamma{(m_{L})}}+p_{N}\frac{\Gamma\left(m_{N},\frac{m_{N}}{\Omega_{N}}\nu_{N}\right)}{\Gamma(m_{N})}\right)^{b}\Bigg],
                %%%%
            \notag\\
                 &\stackrel{}{=}\mathbb{E}_{l}\Bigg[(\bar{\mathcal{A}}_{2,L}(l_{2,L})+\bar{\mathcal{A}}_{2,N}(l_{2,N}))\times\left(p_{L} \frac{\Gamma\left(m_{L},\frac{m_{L}}{\Omega_{L}}\nu_{L}\right)}{\Gamma{(m_{L})}}+p_{N}\frac{\Gamma\left(m_{N},\frac{m_{N}}{\Omega_{N}}\nu_{N}\right)}{\Gamma(m_{N})}\right)^{b}\Bigg],
       %%%%
                 \notag\\
                 &\stackrel{(b)}{=}\mathbb{E}_{l}\Bigg[\bar{\mathcal{A}}_{2,L}(l_{2,L})\left(p_{L} \frac{\Gamma\left(m_{L},\frac{m_{L}}{\Omega_{L}}\nu_{L}\right)}{\Gamma{(m_{L})}}\right)^{b}\Bigg]+\mathbb{E}_{l}\Bigg[\bar{\mathcal{A}}_{2,N}(l_{2,N})\left( p_{N}\frac{\Gamma\left(m_{N},\frac{m_{N}}{\Omega_{N}}\nu_{N}\right)}{\Gamma(m_{N})}\right)^{b}\Bigg],
                 %%%
                 \notag\\&\stackrel{(c)}=\mathbb{E}_{l}\Bigg[\bar{\mathcal{A}}_{2,L}(l_{2,L})\,p_{L}^{b}
                 \left(1-\frac{\gamma\left(m_{L},\frac{m_{L}}{\Omega_{L}}\nu_{L}\right)}{\Gamma{(m_{L})}}\right)^{b} \Bigg] +\mathbb{E}_{l}\Bigg[\bar{\mathcal{A}}_{2,N}(l_{2,N})\, p_{N}^{b}\left(1-\frac{\gamma\left(m_{N},\frac{m_{N}}{\Omega_{N}}\nu_{N}\right)}{\Gamma(m_{N})}\right)^{b} \Bigg],
                 %%%%
                 \notag\\
                 &\stackrel{(d)}{\approx} \mathbb{E}_{l}\Bigg[\bar{\mathcal{A}}_{2,L}(l_{2,L}) p_{L}^{b}
                 \left(1-[1-e^{-\zeta_{L}\nu_{L}}]^{m_{L}}\right)^{b} \Bigg]+\mathbb{E}_{l}\Bigg[\bar{\mathcal{A}}_{2,N}(l_{2,N}) \,p_{N}^{b}\left(1-[1-e^{-\zeta_{N}\nu_{N}}]^{m_{N}}\right)^{b} \Bigg],
                 %%%%
                 \notag\\
                 &\stackrel{(e)}=\mathbb{E}_{l}\Bigg[\bar{\mathcal{A}}_{2,L}(l_{2,L})p_{L}^{b}
             \sum_{k=0}^{b}\binom{b}{k}\left(-[1-e^{-\zeta_{L}\nu_{L}}]^{m_{L}}\right)^{k} \Bigg]+\mathbb{E}_{l}\Bigg[\bar{\mathcal{A}}_{2,N}(l_{2,N})\,p_{N}^{b}\sum_{k=0}^{b}\binom{b}{k}\left(-[1-e^{-\zeta_{N}\nu_{N}}]^{m_{N}}\right)^{k} \Bigg],
                 %%%%
                 \notag\\
                 & \stackrel{(f)}{=}\mathbb{E}_{l}\Bigg[\bar{\mathcal{A}}_{2,L}(l_{2,L})p_{L}^{b}
                 \sum_{k=0}^{b}\sum_{\ddot{k}=0}^{m_{L}k} \binom{b}{k}\binom{m_{L}k}{\ddot{k}}(-1)^{\ddot{k} +k}e^{-\zeta_{L}\nu_{L}\ddot{k}} \Bigg]
                 +  \mathbb{E}_{l}\Bigg[\bar{\mathcal{A}}_{2,N}(l_{2,N}) p_{N}^{b}\sum_{k=0}^{b}\sum_{\ddot{k}=0}^{m_{N}k} \binom{b}{k}\binom{m_{N}k}{\ddot{k}}(-1)^{\ddot{k}+k}e^{-\zeta_{N}\nu_{N}\ddot{k}}\Bigg],
                %\notag\\
               %  &\stackrel{(g)}=
                % \sum_{k=0}^{b}\binom{b}{k}(-1)^{k}\left(p_{L}^{b}\sum_{\ddot{k}=0}^{m_{L}k}\binom{m_{L}k}{\ddot{k}}(-1)^{\ddot{k}}\int_{0}^{d^{\alpha_{2,L}}}e^{-\zeta_{L}\ddot{k}\ddot{\nu_{L}} l_{\text{2,L}}} G_{1}(l_{2,L})dl_{2,L}+\notag\right.\\
                 %&\left.\quad\quad\quad\quad\quad\quad\quad\quad\quad\quad\quad\quad\quad\quad\quad\quad\quad\quad\quad  p_{N}^{b}\sum_{\ddot{k}=0}^{m_{N}k}\binom{m_{N}k}{\ddot{k}}(-1)^{\ddot{k}}\int_{d^{\alpha_{2,N}}}^{\infty}e^{-\zeta_{N}\ddot{k}\ddot{\nu_{N}} l_{\text{2,N}}} G_{2}(l_{2,N})dl_{2,N}\right),
                 \nonumber
          \end{align}
          \endgroup
          \normalsize
where (a) follows from substituting the value of $P_{s,2,L}(\theta_{\mathcal{D}})$ and  $P_{s,2,N}(\theta_{\mathcal{D}})$ from Eq. (\ref{MBS-OR-SBS-LOS}) and Eq. (\ref{MBS-OR-SBS-NLOS}), respectively,  (b) follows from $l_{\text{2,L}}=r_{2,\mathcal{D}}^{\alpha_{2,L}}$ and $l_{\text{2,N}}=r_{2,\mathcal{D}}^{\alpha_{2,N}}$ and the considered blockage model where $p_{L}=1$ when mm-wave intended link distance $r_{2,\mathcal{D}}<d$ and $p_{N}=1$ when mm-wave intended link distance $r_{2,\mathcal{D}}>d$, (c) follows from  $\Gamma(s)=\gamma(s,x)+\Gamma(s,x)$, (d) follows from the CDF of Gamma random variable  which can be tightly upper bounded by $\frac{\gamma\left(m_{L},\frac{m_{L}}{\Omega_{L}}\nu_{L}\right)}{\Gamma{(m_{L})}}<[1-e^{-\zeta_{L}\nu_{L}}]^{m_{L}}$ \cite{alzer1997some}, where $\zeta_{L}\overset{\Delta}{=}m_{L}(m_{L}!)^{-1/m_{L}}$, $\nu_{L}\overset{\Delta}{=}\frac{\theta_{\mathcal{D}}r_{2,\mathcal{D}}^{\alpha_{2,L}}\sigma_{2}^{2}}{P_{2}G_{2}}$, $\zeta_{N}\overset{\Delta}{=}m_{N}(m_{N}!)^{-1/m_{N}}$, and $\nu_{N}\overset{\Delta}{=}\frac{\theta_{\mathcal{D}}r_{2,\mathcal{D}}^{\alpha_{2,N}}\sigma_{2}^{2}}{P_{2}G_{2}}$ \cite{bai2015coverage}. The steps in (e) and (f) are done by following the binomial expansion theorem. Finally, the Lemma~4 follows from de-conditioning on $l$ and using the definitions $\ddot{\nu_{L}}\overset{\Delta}{=}\frac{\nu_{L}}{r_{2,\mathcal{D}}^{\alpha_{2,L}}}=\frac{\nu_{L}}{l_{\text{2,L}}}=\frac{\theta_{\mathcal{D}}\sigma_{2}^{2}}{P_{2}G_{2}}$  and $\ddot{\nu_{N}}\overset{\Delta}{=}\frac{\nu_{N}}{r_{2,\mathcal{D}}^{\alpha_{2,N}}}=\frac{\nu_{N}}{l_{\text{2,N}}}=\frac{\theta_{\mathcal{D}}\sigma_{2}^{2}}{P_{2}G_{2}}$.

\section*{Appendix~D: Proof of Lemma \ref{4}}\label{prof_4}
The $b^{\mathrm{th}}$ moment of the CSP of a typical {\color{black}device} when associated to $\mu$wave MBS  is derived as follows:
\small
\begingroup \allowdisplaybreaks
      \begin{align*}
           &M_{b,\text{1}}(\theta_{\mathcal{D}})=\mathbb{E}_{l_{1}}\bigg[\underbrace{\mathbb{P}(n=1|L_{1,min}=l_{1})}_{\bar{\mathcal{A}}_{1}(l_1)}P_{s,1}(\theta_{\mathcal{D}})^{b}\bigg]\stackrel{(a)}{=}\mathbb{E}_{l_{1}}\bigg[\bar{\mathcal{A}}_{1}(l_1)
                 \prod_{\mathbf{y}_{1,i}\in\mathbf{\Phi}_1\backslash \{\mathbf{y}_{1,0}\}}\frac{1}{\left(1+\theta_{\mathcal{D}}\left(\frac{r_{1,\mathcal{D}}}{\|\mathbf{y}_{1,i}\|}\right)^{\alpha_{1}}\right)^{b}}\bigg],
                 %%%%%%%%
                 \notag\\
                 &\stackrel{(b)}{=}\mathbb{E}_{l_{1}}\bigg[\bar{\mathcal{A}}_{1}(l_1)
                 \exp
                 \left(\int_{r}^{\infty}-2\lambda_{1}\pi\bigg[1-\frac{1}{\left(1+\theta_{\mathcal{D}}\left(\frac{r}{y}\right)^{\alpha_{1}}\right)^{b}}\bigg]y\text{d}y\right)\bigg]\stackrel{(c)}{=}\mathbb{E}_{l_{1}}\bigg[\bar{\mathcal{A}}_{1}(l_1)
                 \exp
                 \left(\int_{l_{1}^{\frac{1}{\alpha_{1}}}}^{\infty}-2\lambda_{1}\pi\bigg[1-\frac{1}{\left(1+\theta_{\mathcal{D}}\frac{l_{1}}{y^{\alpha_{1}}}\right)^{b}}\bigg]y\text{d}y\right),
                 %%%%%
                 \notag\\
                  &\stackrel{(d)}{=} \mathbb{E}_{l_{1}}\bigg[\bar{\mathcal{A}}_{1}(l_1)\exp
                 \left(\int_{0}^{1}-2\lambda_{1}\pi\bigg[1-\frac{1}{\left(1+\theta_{\mathcal{D}}v\right)^{b}}\bigg]v^{-1}\frac{y^{2}}{\alpha_{1}}\text{d}v\right)\bigg],
                 %%%%
                 \notag\\
                 &\stackrel{(e)}{=}\mathbb{E}_{l_{1}}\bigg[\bar{\mathcal{A}}_{1}(l_1)\exp
                 \left(\frac{-2\lambda_{1}\pi l_{1}^{\frac{2}{\alpha_{1}}}}{\alpha_{1}}\int_{0}^{1}\bigg[1-\frac{1}{\left(1+\theta_{\mathcal{D}}v\right)^{b}}\bigg]\frac{1}{v^{\frac{2}{\alpha_{1}}+1}}\text{d}v\right)\Bigg],\notag
          \end{align*}
          \endgroup
          \normalsize
          where (a) follows from taking expectation over $l_{1}=r^{\alpha_{1}}$ and considering the conditional association probability for the typical {\color{black}device} connecting to the MBSs tier given in Lemma~(\ref{1}) and substituting the value of $P_{s,1}(\theta_{\mathcal{D}})$ from Eq. (\ref{MBS-OR-SBS}). In step (b) we apply PGFL of the PPP \cite[Chapter 4]{haenggi2012book}. Step (c) follows from averaging over $l_{1}$. In step (d), we use the change of variable $v=\frac{l_{1}}{y^{\alpha_{1}}}$, $dy=\frac{-1}{\alpha_{1}l_{1}y^{-\alpha_{1}-1}}dv=\frac{-1}{\alpha_{1}}v^{-1}ydv$, when $y=l_{1}^{\frac{1}{\alpha_{1}}}$ $\rightarrow$ $v=1$ and when $y=\infty$ $\rightarrow$ $v=0$ and we swap the integral limits and multiply by $-1$, (e) follows from $y^{2}=l_{1}^{\frac{2}{\alpha_{1}}}/v^{\frac{2}{\alpha_{1}}}$ and doing some mathematical manipulations.
\section*{Appendix~E: Proof of Lemma \ref{5}}\label{prof_5}
While taking the association biases effect in consideration, the $b^{\mathrm{th}}$ moment of the CSP $P_{s,k}(\theta_{\mathcal{D}})$ of the typical {\color{black}device} when it is served by the $k^{\mathrm{th}}$ tier is given as follows:
\begingroup \allowdisplaybreaks
\small
\begin{align}
           &M_{b,\text{k'}}(\theta_{\mathcal{D}})=\mathbb{E}_{r_{k,\mathcal{D}}}\bigg[\mathbb{P}(n=k|r_{k,\mathcal{D}})P_{s,k'}(\theta_{\mathcal{D}})^{b}\bigg],\notag\\
                 &\stackrel{(a)}{=}\mathbb{E}_{r_{k,\mathcal{D}}}\bigg[\prod_{j\neq k}e^{-\pi\lambda_{j}(\hat{P}_{jk}\hat{B}_{jk})^{2/\alpha_{j}}r^{2}}\times\prod_{\mathbf{y}_{k,i}\in\mathbf{\Phi}_k\backslash \{\mathbf{y}_{k,0}\}}\frac{1}{\left(1+\theta_{\mathcal{D}}\left(\frac{r_{k,\mathcal{D}}}{\|\mathbf{y}_{k,i}\|}\right)^{\alpha_{k}}\right)^{b}}\bigg],\notag\\
                 &\stackrel{(b)}{=}\mathbb{E}_{r_{k,\mathcal{D}}}\bigg[\prod_{j\neq k}e^{-\pi\lambda_{j}(\hat{P}_{jk}\hat{B}_{jk})^{2/\alpha_{j}}r^{2}}\times\exp
                 \left(\int_{r_{k,\mathcal{D}}}^{\infty}-2\lambda_{k}\pi\bigg[1-\frac{1}{\left(1+\theta_{\mathcal{D}}\left(\frac{r_{k,\mathcal{D}}}{y}\right)^{\alpha_{k}}\right)^{b}}\bigg]y\text{d}y\right)\bigg],\notag\\
                 &\stackrel{(c)}{=}\int_{0}^{\infty}2\lambda_{k}\pi r e^{-\lambda_{k}\pi r^{2}}e^{-\sum\limits_{j\neq k}\lambda_{j}(\hat{P}_{jk}\hat{B}_{jk})^{2/\alpha_{j}}\pi r^{2}}\times\exp\left(\int_{r}^{\infty}-2\lambda_{k}\pi\bigg[1-\frac{1}{\left(1+\theta_{\mathcal{D}}\left(\frac{r}{y}\right)^{\alpha_{k}}\right)^{b}}\bigg]y\text{d}y\right)\text{d}r,\notag\\
                 &\stackrel{(d)}{=}\int_{0}^{\infty}e^{-q} e^{-q\sum\limits_{j\neq k}\hat{\lambda}_{jk}(\hat{P}_{jk}\hat{B}_{jk})^{2/\alpha_{j}}}\times\exp\left(-2q\int_{0}^{1}\bigg[1-\frac{1}{\left(1+\theta_{\mathcal{D}}v^{\alpha_{k}}\right)^{b}}\bigg]v^{-3}\text{d}v\right)\text{d}q,\notag\\
                 &\stackrel{(e)}{=}\int_{0}^{\infty}e^{-q} e^{-q\sum\limits_{j\neq k}\hat{\lambda}_{jk}(\hat{P}_{jk}\hat{B}_{jk})^{2/\alpha_{j}}}\times\exp\left(-q\int_{1}^{\infty}\bigg[1-\frac{1}{\left(1+\theta_{\mathcal{D}}u^{-\alpha_{k}/2}\right)^{b}}\bigg]\text{d}u\right)\text{d}q,\notag\\
                 &\stackrel{(f)}{=}\int_{0}^{\infty}e^{-q} e^{-q\sum\limits_{j\neq k}\hat{\lambda}_{jk}(\hat{P}_{jk}\hat{B}_{jk})^{2/\alpha_{j}}}\times\exp\left(-q\bigg[\text{ }_{2}F_{1}(b,-\frac{2}{\alpha_{k}};1-\frac{2}{\alpha_{k}};-\theta_{\mathcal{D}})-1\bigg]\right)\text{d}q,\notag\\
                 %&=\int_{0}^{\infty} e^{-q(\sum\limits_{j\neq k}\hat{\lambda}_{jk}(\hat{P}_{jk}\hat{B}_{jk})^{2/\alpha_{j}}+\text{ }_{2}F_{1}(b,-\frac{2}{\alpha_{k}};1-\frac{2}{\alpha_{k}};-\theta_{\mathcal{D}}))}\text{d}q
                 &=\frac{1}{\sum\limits_{j\neq k}\hat{\lambda}_{jk}(\hat{P}_{jk}\hat{B}_{jk})^{2/\alpha_{j}}+\text{ }_{2}F_{1}(b,-\frac{2}{\alpha_{k}};1-\frac{2}{\alpha_{k}};-\theta_{\mathcal{D}})}.\notag
          \end{align}
          \normalsize
          \endgroup
where (a) follows from considering the conditional association probability for the typical {\color{black}device} connecting to the $k^{\mathrm{th}}$ tier given in Eq. \eqref{condition-access}. In step (b), we apply PGFL of the PPP \cite[Chapter 4]{haenggi2012book}. Step (c) follows from averaging over $r_{k,\mathcal{D}}$, step (d) is by using variable substitution $q=\pi\lambda_{k}r^{2}$ and $v=r/y$. In step (e), we perform variable substitution $v=u(\hat{P}_{jk}\hat{B}_{jk})^{-1/\alpha_{j}}$ and step (f) follows from the fact that $\text{ }_{2}F_{1}(b,-\frac{2}{\alpha};1-\frac{2}{\alpha};-\theta)\equiv 1+\int_{1}^{\infty}(1-\frac{1}{(1+\theta h^{-\alpha/2})^{b}})\text{d}h$.
%\appendices

%\section{Proof of Lemma \ref{pdf_distance}}
%\label{prof_3}
%
%\section{Proof of Lemma \ref{average_trans_rate_three_types}}
%%%============================================================
%\label{Average transmission rate_prof}

% Can use something like this to put references on a page
% by themselves when using endfloat and the captionsoff option.

%\bibitem{IEEEhowto:kopka}
%H.~Kopka and P.~W. Daly, \emph{A Guide to \LaTeX}, 3rd~ed.\hskip 1em plus
%  0.5em minus 0.4em\relax Harlow, England: Addison-Wesley, 1999.
%
%\end{thebibliography}

\end{document}